\documentclass[11pt]{article}
\usepackage{graphicx}

\usepackage[colorlinks=true,citecolor=blue]{hyperref}
\usepackage{natbib}
\usepackage{graphicx}
\usepackage{amsfonts}
\usepackage{amsmath}
\usepackage{amssymb}
\usepackage{url}
\usepackage{fancyhdr}
\usepackage{indentfirst}
\usepackage{enumerate}
\usepackage{titlesec}
\usepackage{amsthm}
\usepackage{dsfont}
\usepackage[misc]{ifsym}
\usepackage{multicol}
\usepackage{blindtext}
\usepackage{geometry}
\usepackage{appendix}
\usepackage{float}
\usepackage{caption}

\theoremstyle{definition}

\newtheorem{example}{Example}

\newtheorem{definition}{Definition}

\theoremstyle{plain}
\newtheorem{theorem}{Theorem}
\newtheorem{lemma}{Lemma}
\newtheorem{proposition}{Proposition}

\theoremstyle{remark}
\newtheorem{remark}{Remark}

\makeatletter
\def\blfootnote{\xdef\@thefnmark{}\@footnotetext}
\makeatother

\usepackage{cases}

\theoremstyle{definition}

\def\N{\mathbb{N}}

\def\p{\mathbb{P}}
\def\E{\mathbb{E}}

\def\R{\mathbb{R}}

\def\d{\mathrm{d}}

\usepackage[onehalfspacing]{setspace}

\usepackage{bm}
\usepackage{tikz-qtree}
\usepackage{tikz}

\def\id{\mathds{1}}

\title{Choquet rank-dependent utility with an exogenous unambiguous source}

\author{
Zachary Van Oosten\thanks{Department of Statistics and Actuarial Science, University of Waterloo,  Canada. \Letter~{\scriptsize\url{zjvanoos@uwaterloo.ca}}} \and Ruodu Wang\thanks{Department of Statistics and Actuarial Science, University of Waterloo,  Canada. \Letter~{\scriptsize\url{wang@uwaterloo.ca}}}}
\date{\today}

\begin{document}
	\maketitle
	\begin{abstract}
We axiomatize the Choquet rank–dependent utility model within a Savage framework with an exogenous source of pure risk. This model is a decision model under ambiguity, serving as a conceptual generalization of the Choquet expected utility model.
The model unifies risk and ambiguity and reduces to the rank-dependent utility for pure risks. Our axiomatization uses two main axioms for biseparable preferences, along with some regularity axioms. A benefit of this axiomatization is that the fairly weak regularity axioms guarantee the existence of matching probabilities. Further, we discuss ambiguity attitudes for the CRDU model. We characterize these attitudes by properties of the associated matching probabilities and show that the supermodularity of the matching probability provides a robust representation. Finally, we show that under an additional Property, this model has a different representation using act-dependent distortion functions.

	\end{abstract}

\noindent\rule{\textwidth}{0.5pt}

\section{Introduction}\label{sec:intro}

The expected utility (EU) model, axiomatized by \cite{vNM47}, serves as the theoretical foundation for much of today's theory of decision-making. The EU model has two well-known limitations. First, the EU model is formulated within the pure-risk setting, where acts are given by lotteries over outcomes. This setting is unrealistic in many situations, as a decision maker (DM) is often not fully aware of the distributions of the acts they are comparing.
Second, even in the pure-risk setting, the EU model fails to account for both conceptual and empirical considerations. For example, the EU model cannot account for the empirically observed Allais paradox \citep{A53}.  

To address concerns associated with the pure-risk setting, the seminal work of \cite{S54} redefined acts as functions from states of the world to outcomes, rather than as lotteries over outcomes. This setting, commonly referred to as the Savage setting, does not assume that the likelihoods of events are well-specified a priori, and thus, there is no objective way to view the acts as lotteries. By imposing a set of axioms on preferences over these acts, \cite{S54} was able to elicit a subjective probability, thereby converting from the Savage setting into the pure-risk setting. Furthermore, the Savage axioms yield the EU model in the elicited pure-risk setting. This model is now known as the subjective expected utility (SEU) model.

{The SEU model remained the prevailing decision model in the Savage setting until the critique of \cite{E61}, which showed that if the DM has asymmetric information regarding the likelihoods of the events, the sure-thing principle (an axiom used in \cite{S54} to derive the SEU model) can fail to be behaviorally feasible. In fact, this critique called into question the behavioral feasibility of probabilistically sophisticated \citep{MS92} decision models:  models that start by representing uncertainty with a subjective probability and then evaluate acts solely based on the lotteries they induce over outcomes.
This foundational critique led to the development of what are now known as ambiguity models.} Some prominent ambiguity models include the Choquet expected utility (CEU) model of \cite{S89}, the maxmin expected utility (MEU) model of \cite{GS89}, the smooth ambiguity model of \cite{KMM05}, and the variational preferences model of \cite{MMR06}.  {All of these models share a common theme: in the absence of ambiguity, i.e., when a probability can describe uncertainty and the DM can operate in the pure-risk setting, they return to the SEU model.}

To address the empirical violations of the EU model within the pure-risk setting, the work of \cite{Q82} introduced the rank-dependent utility (RDU) model. This model was a generalization of the EU model that could account for the Allais paradox. Further advancements came with the prospect theory \citep{KT79} and the cumulative prospect theory \citep{TK92}, a generalization of the RDU model. This model was more behaviorally realistic than the RDU model, as it accounted for framing effects and asymmetries in the valuation of gains and losses. For a comprehensive treatment of the prospect theory, see \cite{W10}.

 We are now at the intersection of two research streams: ambiguity models, which reduce to the SEU model in the case of pure risk, and rank-dependent models, which are in the pure-risk setting and do not account for ambiguity. However, some work has been done to bridge this gap. For example, the framework of distributionally robust optimization extends the MEU model by allowing for the reduction to general pure-risk models in the case of no ambiguity; see \cite{DKW19} and \cite{FLW24} for approaches along this line through the lens of risk measures. A different generalization, and the main focus of this paper, is a conceptual extension of the CEU model that reduces to the RDU model under pure risk. This framework, originally introduced in \cite{TF95} as the two-stage approach, retains the functional form of the CEU model but incorporates an additional distortion outside the capacity to capture how the DM distorts probabilities in the case of pure risk. In this paper, we will refer to this as the Choquet rank-dependent utility (CRDU) model. 

As the CRDU model is an instance of the biseperable model (\citealp{GM01}), we provide a biseperable axiomatization of the CRDU model. Afterwards, we present a novel analysis of the ambiguity attitudes associated with this model. Recently, \cite{B25} proposed an axiomatization of the CRDU model. Our approach begins from a fundamentally different modeling assumption. The motivation for this choice comes from the empirical literature validating the behavioral feasibility of the CRDU model; see, for example, \cite{AVW05}, \cite{DKW16}, and \cite{ABKL21}. The empirical literature frequently employs the concept of a matching probability (\citealp{BHSW18}) to determine the value of the capacity under the distortion for different events. To properly define matching probabilities, we need to have an objective notion of pure risk. Furthermore, our axiomatization relies on regularity axioms that allow for the independent verification of the existence of matching probabilities.

In this paper, we adopt the Savage setting, where acts are viewed as mappings from states of the world to real-valued outcomes.\footnote{The setting of \cite{S54} is, in general, not limited to real-valued acts; see \cite{AW20}. We take the perspective of real-valued acts in this paper.} {To define matching probabilities, we will assume the DM has an exogenous source of unambiguous events. This source will be equipped with a probability measure identifying the likelihoods of these events.  Our main regularity axiom is that the decision maker is probabilistically sophisticated regarding acts that depend on events from the exogenous source. Models following this axiom has been studied recently by \cite{BBLW23} in general decision models and by \cite{SVW24} in risk assessment.}

We begin in Section \ref{sec:def} by defining notation and reviewing the mathematical tools that will be used throughout the paper. In Section \ref{sec:set}, we will define the Savage setting with pure risk and discuss how the CEU and CRDU models fit into this setting. In Section \ref{sec:ax1}, we present the axioms for our axiomatization of the CRDU model and derive a representation in Theorem \ref{theo:good}. To properly discuss how the ambiguity component of the preferences is incorporated into this representation, we relate an aspect of the representation to the concept of matching probabilities. Section \ref{sec:prop} focuses on analyzing various ambiguity attitudes. We start in Section \ref{sec:ambiguityaversion} by defining comparative ambiguity attitudes and deriving their representation through matching probabilities. Afterwards, we compare our definition to the comparative ambiguity attitudes of \cite{E04}. We finish the section by defining absolute ambiguity attitudes once a normalization for ambiguity neutrality is fixed. In Section \ref{sec:uncertaintyaversion}, we discuss the propensity for diversification. We compare the consequences of this preference in the CRDU model with those in the restricted CEU model. Finally, in Section \ref{sec:robust}, we show how some of the ideas from Section \ref{sec:uncertaintyaversion} give rise to a robust representation, analogous to the way that propensity for diversification connects the CEU and MEU models. In Section \ref{sec:alt},  {we adopt the setting of \cite{SVW24} and assume that the DM has a reference probability measure on the whole space whose restriction agrees with the probability measure on the unambiguous source. This allows us to define and characterize a natural form of ambiguity aversion that we call reference ambiguity aversion.} We finish the section by providing an alternative representation for the CRDU model using act-dependent distortion functions in Theorem \ref{theo:2}, offering a different view on decision-making under ambiguity: both pure risks and ambiguous acts are distorted. We finish by characterizing an ambiguity model that reduces to the dual model of \cite{Y87} for pure risks,
which has a concise axiomatic formulation. 
We conclude the paper in Section \ref{sec:con}. For completeness, the proofs of all results, along with related additional results, are provided in the appendices.

\section{Preliminaries}

\label{sec:Prelim}

\subsection{Definitions and notation}

\label{sec:def}

Throughout the paper, we consider a measurable space $(\Omega, \mathcal{F})$ representing the future states of the world. A \emph{capacity}  is a function $\nu:\mathcal F\mapsto \R$ that satisfies  $\nu(\varnothing)=0$, $\nu(\Omega)=1$, and  $\nu(A)\leq \nu(B)$ for all $A,B\in \mathcal{F}$ with $A\subseteq B.$ In some literature (e.g., \citealp{MM04}), a capacity is not assumed to satisfy $\nu(\Omega)=1$, but this difference is not essential. A capacity $\nu$ is said to be: \emph{upward continuous} if for all increasing sequences $(A_n)_{n\in\N}\subseteq\mathcal{F}$ it holds that $\lim_{n\to\infty}\nu(A_n)=\nu\left(\bigcup_{n=1}^\infty A\right)$, \emph{downward continuous} if for all decreasing sequences $(A_n)_{n\in \N}\subseteq\mathcal{F}$ it holds that $\lim_{n\to\infty}\nu(A_n)=\nu\left(\bigcap_{n=1}^\infty A_n\right)$, and \emph{continuous} if it is both upward continuous and downward continuous.
A capacity is called a \emph{probability measure} if it is countably additive. 
We denote the set of probability measures  on $(\Omega,\mathcal F)$ by $\mathcal{M}_1$. 

 {As in the Savage setting,} an \emph{act} is defined as a bounded real-valued measurable function. We will denote the space of acts by $\mathcal{X}$. Constant acts in $\mathcal X$ are identified with constants in $\R$. 
Given $X\in \mathcal{X}$, we write $\|X\|=\sup_{\omega\in \Omega}|X(\omega)|$. A set $\mathcal{S}\subseteq \mathcal{X}$ is said to be \emph{bounded} if $\sup_{X\in \mathcal{S}}\|X\|<\infty$ and \emph{B-closed} if, for any bounded $(X_n)_{n\in N}\subseteq\mathcal{S}$ that pointwise converges to some  $X\in \mathcal{X}$, it holds that $X\in \mathcal{S}$.

For $X,Y\in \mathcal{X}$,  $X\geq Y$ means $X(\omega)\geq Y(\omega)$ for all $\omega \in \Omega$. Given $\mu\in \mathcal{M}_1$ and $X,Y\in \mathcal{X}$, we say that $X\geq_{\mu}^{\mathrm{as}} Y$ if $\mu(\{X\geq Y\})=1$, and $X\geq^{\mathrm{sd}}_{\mu}Y$ if $\mu(X>x)\geq \mu(Y>x)$ for all $x\in \mathbb{R}$. Clearly,
$
X\ge Y ~\Longrightarrow~ X \ge ^{\mathrm{as}}_{\mu} Y ~\Longrightarrow~ X \ge ^{\mathrm{sd}}_{\mu}  Y.$ 

For a continuous capacity $\nu$, the \emph{core of $\nu$}, denoted by $\mathfrak{C}(\nu)$, is defined as 
$$\mathfrak{C}(\nu)=\left\{\mu\in \mathcal{M}_{1}: \mu(A)\geq \nu(A)\text{ for all }A\in \mathcal{F}\right\}.$$   
The following are common properties of a continuous capacity $\nu$.
\begin{itemize}\setlength{\itemsep}{1pt}
    \item[] Supermodularity: For all $A,B\in \mathcal{F}$, $\nu(A)+\nu(B)\leq \nu(A\cup B)+\nu(A\cap B)$.
    \item[] Submodularity: For all $A,B\in \mathcal{F}$, $\nu(A)+\nu(B)\geq \nu(A\cup B)+\nu(A\cap B)$.
    \item[] Balanced:  $\mathfrak{C}(\nu)\neq \varnothing$.
    \item[] Exactness: For all $A\in \mathcal{F}$, $\nu(A)=\min_{\mu\in \mathfrak{C}(\nu)}\mu(A)$.
\end{itemize}
For more on the core and the previous properties of capacities, see \cite{D94} and \cite{MM04}. 

Given a capacity $\nu$, the \emph{Choquet integral} with respect to $\nu$ is given by
$$\int X\d\nu=\int_0^\infty \nu(X>x)\d x+\int_{-\infty}^0\left(\nu(X>x)-1\right)\d x,~~X\in \mathcal{X}.$$ 
We say that $X,Y\in \mathcal{X}$ are \emph{comonotonic} if for all $\omega,\tilde{\omega}\in \Omega,$
$(X(\omega)-X(\tilde\omega))(Y(\omega)-Y(\tilde{\omega}))\geq 0.$ 
As shown by \cite{S86},   a function $V:\mathcal{X}\to\mathbb{R}$ can be represented as $V(X)=\int_{\Omega}X \d\nu$ for some capacity $\nu$ if and only if $V$ is monotone ($V(X)\geq V(Y)$ whenever $X\geq Y$), comonotonic additive ($V(X+Y)=V(X)+V(Y)$ for comonotonic $X,Y$), and normalized ($V(1)=1$).  An increasing (always in the non-strict sense) function $g:[0,1]\to[0,1]$ with $g(0)=0$ and $g(1)=1$ is called a \emph{distortion function}. {A \emph{von Neumann--Morgenstern (vNM) utility function} is a non-constant increasing function $u:\mathbb{R}\to\mathbb{R}$.}

\subsection{The Savage setting with pure risk and Choquet rank-dependent utility}

\label{sec:set}

We will work in
a refined Savage setting, as described below. Fix a sub-$\sigma$-algebra $\mathcal{G}$, {which is to be interpreted as an exogenous source of unambiguous events. As the events in $\mathcal{G}$ are unambiguous, the DM is able to assign likelihoods to them consistent with a probability measure. Let $\mathbb{P}:\mathcal{G}\to[0,1]$ denote the probability measure assigning these likelihoods. As common in risk management, we will assume that $\mathbb{P}$ is an atomless; that is, $\{\mathbb{P}(B) : B \in \mathcal{G},~B \subseteq A\}=[0, \mathbb{P}(A)] $ for every $A \in \mathcal{G}$.} We refer to this setting as the \emph{Savage setting with pure risk}. We call $\mathbb{P}$ the \emph{exogenous probability measure}.

A preference relation $\succsim$ representing the DM's preferences is a total preorder on $\mathcal{X}$; that is, (i) for all $X,Y\in \mathcal{X}$, $X\succsim Y$ or $Y\succsim X$; (ii) for all $X,Y,Z\in \mathcal{X}$, if $X\succsim Y$ and $Y\succsim Z$ then $X\succsim Z$.
We write $\mathcal{X}(\mathcal{G})$ for the subspace of $\mathcal{X}$ consisting of $\mathcal{G}$-measurable acts. Under our setting, acts $X\in \mathcal{X}(\mathcal{G})$ exhibit no distributional ambiguity, and the preference relation $\succsim$ restricted to $\mathcal{X}(\mathcal{G})$ should be determined solely by the distributions of these acts under the exogenous probability measure $\mathbb{P}$. This brings us to our first axiom of risk conformity.\footnote{Risk conformity was called differently  by \cite{SVW24} as  partial law invariance.} 
\begin{enumerate}[(SRM)]\setlength{\itemsep}{1pt}
    \item[(RC)] Risk conformity: For all $X,Y\in \mathcal{X}(\mathcal{G})$ satisfying $X=_{\mathbb{P}}^{\mathrm{sd}}Y$, $X\simeq Y$.
\end{enumerate}
If the preference relation $\succsim$ satisfies Axiom (RC), then the restriction of $\succsim$ to $\mathcal{X}(\mathcal{G})$, called the \emph{risk preferences} of $\succsim$, corresponds to a preference relation over lotteries. To evaluate general acts, i.e., acts not necessarily in $\mathcal{X}(\mathcal{G})$, the DM must consider both their risk preferences and their perception and attitudes regarding ambiguity.  
In the discussions later, the \emph{risk component} of $\succsim$   refers to the influence of the risk preferences on the evaluation of general acts, and  the \emph{ambiguity component} of $\succsim$   refers to the influence of ambiguity on the evaluation of general acts.

A capacity $\nu$
is \emph{risk conforming} if  $\nu|_{\mathcal{G}}=\mathbb{P}$.
Risk conformity means that, for unambiguous events, the capacity behaves like the exogenous probability measure $\p$. 
When $\nu$  {describes the DM's subjective evaluation regarding the likelihoods of events, risk conformity is essential, as the DM believes that $\mathbb{P}$ correctly specifies the likelihood of unambiguous events.}

The CEU model (\citealp{S89}) was originally formulated in the setting of \cite{AA63}. We now explain the {relevant}  formulation of the CEU model in the Savage setting with pure risk.\footnote{The choice below to take $u$ to be continuous and strictly increasing, as well as $\nu$ to be continuous, is not required in general, but this choice is adopted for reasons that will become apparent in the next section.} The DM has a continuous strictly increasing vNM utility $u$ which describes their evaluation of outcomes. To describe the ambiguity component of the preferences, the DM has a continuous risk-confirming capacity $\nu$ that can be interpreted as their  {subjective assessment of the likelihood of events}. The preference relation $\succsim$ of the DM in the CEU model is given by
\begin{equation}
    \label{eq:sw}
    X\succsim Y\iff \int u(X)\d\nu\geq \int u(Y)\d\nu,~~X,Y\in \mathcal{X}.
\end{equation}
It is well known that, given a preference relation $\succsim$ satisfying \eqref{eq:sw}, $u$ is unique up to positive affine transformations and $\nu$ is unique. 
{Since $\nu$ in \eqref{eq:sw} is risk conforming}, it is straightforward to show that a CEU preference relation  $\succsim$ satisfies Axiom (RC) and the risk preferences for $\succsim$ are given by the EU model. As noted in the introduction, the majority of ambiguity models yield the EU model for risk preferences. A key conceptual assumption of the CEU model is that if the DM believes a probability measure $\mu\in\mathcal{M}_1$ uniquely describes their {subjective assessment of the likelihood of events}, i.e., they perceive no ambiguity, then they take $\nu=\mu$ in \eqref{eq:sw} and their preference relation $\succsim$ is given by the SEU model. The work of \cite{SW92} provides an axiomatization of CEU preference relations, but rather than assuming that 
$\mathbb{P}$ is given exogenously they incorporate its derivation into the axiomatization. 

 {A conceptual generalization of the CEU model is the CRDU model, first considered by \cite{TF95}. We now explain the relevant formulation of the CRDU model in the Savage setting with pure risk. 
Once again, the DM has a continuous strictly increasing vNM utility $u$ and a continuous risk-confirming capacity $\nu$, which serve the same purpose as the CEU model. However, unlike the CEU model, the DM has an additional element describing the risk component of her preferences. This element is given by a continuous strictly increasing distortion function $g$, which describes how the DM distorts probabilities. The preference relation $\succsim$ of the DM in the CRDU model is given by
\begin{equation}
    \label{eq:rep}
    X\succsim Y\iff \int u(X)\d(g\circ\nu)\geq \int u(Y)\d(g\circ\nu),~~X,Y\in \mathcal{X}.
\end{equation}
Since any preference relation satisfying \eqref{eq:rep} is biseparable (\citealp{GM01}), and $\nu$ must be risk conforming, it is straightforward to show that for any such preference relation, the utility function $u$ is unique up to positive affine transformations, $g$ is unique, and $\nu$ is unique.}

 {Given a CRDU preference relation $\succsim$ in \eqref{eq:rep}, since $\nu$ is risk conforming, we have
 $$X\succsim Y\iff \int_\Omega u(X)\d(g\circ \mathbb{P})\geq \int_\Omega u(Y)\d(g\circ \mathbb{P}), ~~X,Y\in \mathcal{X}(\mathcal{G}).$$ 
Therefore, $\succsim$ satisfies Axiom (RC) and the risk preferences are given by the RDU model of \cite{Q82}. This explains the terminology ``Choquet rank-dependent utility'', as the model has a similar form as the CEU model but allows for risk preferences given by the RDU model. Furthermore, if there exists $x\in [0,1]$ such that $g(x)\neq x$, then the risk preferences are not an instance of the EU model.}  

An axiomatization for the CRDU model was recently proposed by \cite{B25} in a setting similar to ours; however, they do not assume that $\mathbb{P}$ is given exogenously. In the next section, we will provide an axiomatization of the CRDU model fundamentally different from \cite{B25}. This axiomatization will consist of the regularity Axiom (RC), three additional regularity axioms that are also natural in the Savage setting with pure risk, and two additional axioms for biseparable preferences. The primary motivation for this axiomatization is that the regularity axioms enable us to show the existence of matching probabilities (\cite{B25} has this as an axiom) and define the necessary tools for establishing the two axioms for biseparable preferences. This highlights the suitability of the Savage setting with pure risk for deriving axiomatizations of non-traditional models of ambiguity within a Savage-like setting. The proofs of the main results, along with several supplementary results, are provided in the Appendix.

\section{Axioms and representation}

\label{sec:ax1}

We now state and discuss several additional axioms for preference relations $\succsim$. Afterwards, in Theorem \ref{theo:good}, we will present the main representation.  
\begin{enumerate}[(SRM)]
\setlength{\itemsep}{1pt}
    \item[(M)] Monotonicity: For all $X,Y\in \mathcal{X}$ satisfying $X\geq Y$, $X\succsim Y$.
\end{enumerate}
Axiom (M) is standard for preference relations in the Savage setting. If the act $X$ pays more than the act $Y$ for every state of the world, then it is natural that the DM would prefer $X$ to $Y$.
\begin{enumerate}[(SRM)]\setlength{\itemsep}{1pt}
    \item[(SRM)] Strict risk monotonicity: For all $X,Y\in \mathcal{X}(\mathcal{G})$ satisfying $X>_{\mathbb{P}}^{\mathrm{as}} Y$, $X\succ Y$.
\end{enumerate}
Axiom (SRM) reflects the fact that the DM believes that the probability measure $\mathbb{P}$ is well-specified on $\mathcal{G}$. Thus, a stronger form of monotonicity, using distributions, holds for acts in $\mathcal{X}(\mathcal{G})$. 

\begin{enumerate}[(SRM)]\setlength{\itemsep}{1pt}
    \item[(C)] Continuity:  For all $X\in \mathcal{X}$, $\{Y\in \mathcal{X}:X\succsim Y\}$ and $\{Y\in \mathcal{X}:Y\succsim X\}$ are B-closed.
\end{enumerate}
Axiom (C) appeared in \cite{CM95} and \cite{CL07} in the pure-risk setting. The use of this axiom in a Savage setting is common in the literature on risk measures, where it is referred to as the Lebesgue property; see \cite{FS16}.


The axioms discussed thus far concern only certain notions of regularity and essentially impose no specific functional forms on the preference relations. The axioms governing specific forms of the risk preferences will be discussed later in the section. {The following proposition confirms that a CRDU preference relation satisfies the above regularity axioms.} 

\begin{proposition}\label{prop:consis}
    Let $\succsim$ be a CRDU preference relation. Then $\succsim$ satisfies Axioms (RC), (M), (SRM), and (C).
\end{proposition}

{An additional confirmation of the suitability of our regularity axioms is that they ensure the existence of matching probabilities, an object we now recall.} For notational simplicity, we will write $A\succsim B$ for $A,B\in \mathcal{F}$ to mean $\id_A\succsim\id_B$, where $\id_A$ denotes the binary act which yields $1$ if $\omega\in A$ and $0$ otherwise.

\begin{definition}
Given a preference relation $\succsim$, a function $\nu:\mathcal{F}\to [0,1]$ is a \emph{$\succsim$-matching probability} if for all $A\in \mathcal{F}$, there exists $R_A\in \mathcal{G}$ such that $\nu(A)=\mathbb{P}(R_A)$ and $A\simeq R_A$.    
\end{definition}

Intuitively, a $\succsim$-matching probability uses the restriction of the preference relation $\succsim$ on $\mathcal{F}$ to produce a description of ambiguity on the level of events. The definition formalizes the idea that, for each event $A \in \mathcal{F}$, the preference relation $\succsim$ is used to identify a unambiguous probability $\nu(A)$ such that a bet paying \$1 if $A$ occurs (and \$0 otherwise) and a unambiguous bet that pays \$1 with probability $\nu(A)$ (and \$0 otherwise) are equally preferred. Our definition of a $\succsim$-matching probability only works in the Savage setting with pure risk, as $\mathcal{G}$ needs to contain unambiguous events to have the aforementioned interpretation. 

\begin{proposition}
    \label{proposition:match}
    For every preference relation $\succsim$ satisfying Axioms (RC), (M), (SRM), and (C), there exists a unique $\succsim$-matching probability $\nu$. Furthermore, $\nu$ is a continuous risk-conforming capacity, and 
    \begin{equation}
    \label{eq:setrep}
        A\succsim B\iff \nu(A)\geq \nu(B),~~A,B\in \mathcal{F}.
    \end{equation}
\end{proposition}

Let $\succsim$ satisfy Axioms (RC), (M), (SRM), and (C). By Proposition \ref{proposition:match}, the $\succsim$-matching probability is risk conforming. Therefore, the $\succsim$-matching probability {correctly identifies the likelihood of unambiguous events specified by  $\mathbb{P}$}. The following shows that the $\succsim$-matching probability can be used to isolate the ambiguity component of a CRDU preference relation, which is conceptually similar to \citet[Theorem 3.1]{DKW16}.

\begin{proposition}\label{prop:matchCheck}
    Let $\succsim$ be a CRDU preference relation in \eqref{eq:rep}, then $\nu$ is the unique $\succsim$-matching probability.
\end{proposition}


{The remaining two axioms are motivated by the fact that a CRDU preference relation is biseperable. These axioms make reference to certain constructions guaranteed to exist from our regularity axioms. For example, both make reference to certainty equivalents.} To show existence, as well as a continuity property, we have the following proposition. 

\begin{proposition}
    \label{prop:cert}
    Given a preference relation $\succsim$ satisfying Axioms (RC), (M), (SRM), and (C), for each $X\in \mathcal{X}$, there exists a unique $c_X\in \mathbb{R}$ such that $X\simeq c_X$. Moreover, if $(X_n)_{n\in \N}\subseteq \mathcal{X}$ is a bounded sequence converging pointwise to $X\in \mathcal{X}$, $\lim_{n\to\infty}c_{X_n}=c_X$.
\end{proposition}

The following axiom, or variations of it, has often been used for the axiomatization of the RDU model. In what follows, we use $\mathcal{G}_{\mathrm{cf}}$ to denote the set $\{R\in \mathcal{G}:\mathbb{P}(R)=1/2\}$ of ``coin flip" events. Given $A\in \mathcal{F}$ and $x,y\in \mathbb{R}$, $xAy$ denotes the binary act which yields $x$ if $\omega\in A$ and $y$ otherwise.

\begin{enumerate}[(SRM)]\setlength{\itemsep}{1pt}
    \item[(RS)] Risk symmetry: For all $R\in\mathcal{G}_{\mathrm{cf}}$ and $x,y,z,z'\in \mathbb{R}$ satisfying $x\geq y$ and $z,z'\in [y,x]$ $$ c_{xRz}Rc_{z'Ry}\simeq c_{xRz'}Rc_{zRy}.$$
\end{enumerate}
It is not hard to show that Axiom (RS) is a weaker version of the weak independence axiom from \cite{Q82}. The formulation of Axiom (RS) is similar to Axiom (A6) in \cite{GMMS01}, where various interpretations of axioms of this type and their historical development are discussed.

To define the final axiom, we need to define a binary operation on $\mathcal{X}$ that can be viewed as a subjective mixture operation for acts. As a preliminary step, following the approach of \cite{GMMS03}, we first introduce a subjective mixture for outcomes. The following proposition demonstrates that, once a non-trivial risky event is fixed, our regularity axioms guarantee both the existence and uniqueness of a subjective mixture for outcomes associated with that event.

\begin{proposition}
    \label{prop:existance}
    Let $\succsim$ satisfy Axioms (RC), (M), (SRM), and (C). Given $R\in \mathcal{G}$ satisfying $\mathbb{P}(A)\in (0,1)$ and $x,y\in \mathbb{R}$ with $x\geq y$, there exists a unique $z\in [y,x]$ such that
    $xRy\simeq c_{xRz}Rc_{zRy}.$
\end{proposition}
If the preference relation $\succsim$ satisfies Axioms (RC), (M), (SRM), and (C), then, 
given $x,y\in \mathbb{R}$, as $x\geq y$ or $y\geq x$, we define $(1/2)x\oplus(1/2)y$ as the unique $z$ given in Proposition \ref{prop:existance} with $R\in \mathcal{G}_{\mathrm{cf}}$, following \cite{GMMS03}. Clearly, by Axiom (RC), this binary operation is independent of which specific set $R\in \mathcal{G}_{\mathrm{cf}}$ we choose. For two acts $X,Y\in \mathcal{X}$, we define the subjective mixture of acts $(1/2)X\oplus(1/2)Y$ by
$$\left[(1/2)X\oplus(1/2)Y\right](\omega)=(1/2)X(\omega)\oplus (1/2)Y(\omega).$$
Proposition \ref{prop:meas} in Appendix \ref{ap:ax1} confirms that $(1/2)X\oplus (1/2)Y\in \mathcal{X}$; that is, $(1/2)X\oplus (1/2)Y$ is bounded and measurable. 

We are now ready to present the final axiom, which is the comonotonic independence axiom of \cite{S89}, but for the subjective mixture of acts. The following axiom assumes that the preference relation $\succsim$ already satisfies Axioms (RC), (M), (SRM), and (C), since the existence of the subjective mixture operation for acts relies on these axioms.

\begin{enumerate}[(SRM)]\setlength{\itemsep}{1pt}
    \item[(SCI)] Subjective comonotonic independence: For all  pairwise comonotonic $X,Y,Z\in \mathcal{X}$,
    $$X\simeq Y\Longrightarrow (1/2)X \oplus (1/2) Z\simeq (1/2) Y \oplus (1/2)Z.$$
\end{enumerate} 
{The following result shows that the above axioms fully characterize CRDU preference relations.}

\begin{theorem}
    \label{theo:good}
    The preference relation $\succsim$ satisfies Axioms (RC), (M), (SRM), (C), (RS), and (SCI) if and only if  $\succsim$ is a CRDU preference relation, that is,
    there exist a continuous strictly increasing vNM utility function $u$, a strictly increasing continuous distortion function $g$, and a continuous risk-conforming capacity $\nu$ such that
 $$    X\succsim Y\iff \int u(X)\d(g\circ\nu)\geq \int u(Y)\d(g\circ\nu),~~X,Y\in \mathcal{X}.$$
\end{theorem}

Given a CRDU preference relation $\succsim$, as the strictly increasing continuous distortion function $g$ in \eqref{eq:rep} is uniquely determined, we will henceforth refer to $g$ as the $\succsim$-distortion. Likewise,  the strictly increasing vNM utility function $u$ in \eqref{eq:rep} is unique up to positive affine transformations, and we will henceforth refer to the version of $u$ satisfying $u(0)=0$ and $u(1)=1$ as the $\succsim$-utility.  

In the Savage setting with pure risk, Axiom (RC) ensures that the $\succsim$-distortion, which represents an aspect of the risk component of $\succsim$, can be identified directly from the preference relation $\succsim$. This, in turn, separates it from the $\succsim$-matching probability, which represents the ambiguity component of $\succsim$. {In the general Savage setting, i.e., the setting where $\mathcal{G}$ is not assumed to be an exogenously given unambiguous source, such a separation of the risk and ambiguity components may not be possible.} The next example illustrates this point with a variational preference model. A similar example can be constructed for RDU and omitted here. 

\begin{example}\label{ex:HS}
  Consider the general Savage setting without specifying $\mathcal G$. Suppose that the DM's preference relation $\succsim$ satisfies 
    \begin{equation}
        \label{eq:example}X\succsim Y\iff \inf_{\mu\in \mathcal M_1}\left\{   \int_{\Omega}X \d\mu  + \beta H(\mu|\mathbb{Q})\right\}\ge \inf_{\mu\in \mathcal M_1}\left\{ \int_{\Omega}Y \d\mu   + \beta H(\mu|\mathbb{Q}) \right\},
    \end{equation}
    where {$\mathbb{Q}\in \mathcal{M}_1$ is a reference probability measure fixed by the DM},  $\beta>0$,  
    and  $H(\mu|\mathbb{Q})$ is the Kullback--Leibler divergence (relative entropy) of $\mu$ from $\mathbb{Q}$.\footnote{If $\mu$ is absolutely continuous with respect to $\mathbb{Q}$, then $H(\mu|\mathbb{Q})=\E^\mu[\log (\d \mu/\d \mathbb{Q})] $, where $\d \mu/\d \mathbb{Q}$ is the Radon–Nikodym derivative; otherwise $H(\mu|\mathbb{Q})=\infty$.}
   The model \eqref{eq:example} is a special case of both the multiplier preferences of \cite{HS01} and the variational preferences of \cite{MMR06}. 
Via the Donsker--Varadhan variational formula,  \eqref{eq:example} can be alternatively expressed by 
  \begin{equation}
        \label{eq:example1-2}X\succsim Y\iff
        \int_{\Omega}-\exp(- X /\beta)\d\mathbb{Q}
        \ge  \int_{\Omega}-\exp(- Y /\beta)\d\mathbb{Q}.
    \end{equation}
The model \eqref{eq:example}--\eqref{eq:example1-2} can be explained by two distinct behavioral considerations.
\begin{enumerate}[(a)]
    \item The DM 
    has the expected-value risk preferences, but she is unsure about the accuracy of the reference probability measure $\mathbb{Q}$. As alternative probability measures $\mu\in \mathcal M_1$ are plausible, the DM  
    uses the relative entropy to penalize the deviation of $\mu$ from $\mathbb{Q}$ and then aggregates these evaluations.  
    \item  The DM is confident about the accuracy of the reference probability measure $\mathbb{Q}$, and she has EU risk preferences with an exponential utility function $  u: x\mapsto -\exp(-x /\beta)$.
\end{enumerate}   
Since both cases lead to the same preference relation, the risk and ambiguity components of $\succsim$ cannot be uniquely identified from the DM's choices over $\mathcal X$.  In case (a), $\beta $ is a parameter of ambiguity aversion and in case (b), it is a parameter of risk aversion. 
With the help of the set $\mathcal X(\mathcal G)$  {in our Savage setting with pure risk}, we are able to correctly separate these two components. 
\end{example}

\section{Ambiguity attitudes}

 {As discussed in the previous section, given a CRDU preference relation $\succsim$, the ambiguity component of $\succsim$ is described by the $\succsim$-matching probability.} Therefore, in this section, we study ambiguity attitudes of CRDU preferences through properties of the matching probabilities. 

\label{sec:prop}

\subsection{Comparative and absolute ambiguity attitudes}

\label{sec:ambiguityaversion}

We will first define comparative ambiguity attitudes for CRDU preference relations. This will then allow us to define absolute ambiguity attitudes once a normalization for ambiguity neutrality is provided. Many of the notions presented here will be similar to those in \cite{E04}, but one key difference exists: \citeauthor{E04}'s (\citeyear{E04}) definition of comparative ambiguity attitudes between two preference relations $\succsim_1$ and $\succsim_2$\footnote{\cite{E04} used the term ``uncertainty" in place of our term ``ambiguity".} implies the equivalence of risk preferences between $\succsim_1$ and $\succsim_2$, whereas our notion is independent of risk preferences. 

Given the CRDU preference relations $\succsim_1$ and $\succsim_2$, we say that $\succsim_2$ is \emph{more ambiguity averse than} $\succsim_1$ if 
$$R\succsim_1 A \implies R\succsim_2 A,~~~~~~R\in \mathcal{G}~\text{and}~A\in \mathcal{F}.$$
The interpretation of this definition is quite natural: 
As events $R\in \mathcal G$ have no ambiguity, they can serve as the benchmark for comparing ambiguous events. 
This is similar to the classic notion of comparative risk aversion (\citealp{P64} and \citealp{Y69}), where constant acts serve as the benchmark. 

We have the following proposition characterizing comparative ambiguity attitudes. 
\begin{theorem}
    \label{prop:comAm}
    Let $\succsim_1$ and $\succsim_2$ be   CRDU preference relations with matching probabilities $\nu_1$ and $\nu_2$ respectively. Then, $\succsim_2$ is more ambiguity averse than $\succsim_1$ if and only if 
    $\nu_1(A)\geq \nu_2(A)$ for all $A\in \mathcal{F}$.
\end{theorem}

In Theorem \ref{prop:comAm}, the $\succsim_1$-distortion and the $\succsim_2$-distortion may not be the same, and similarly, the two utility functions may not be the same.
This is consistent with the fact that the ambiguity component in the CRDU model is separated from the risk component. 
To get a full comparative implication on all acts, instead of only binary ones, we need to force the risk components of both preference relations to be identical, as shown in the next result. 

\begin{proposition}
    \label{prop:attitudes}
    Let $\succsim_1$ and $\succsim_2$ be CRDU preference relations. Then, the comparative implication  \begin{equation}
        \label{eq:ep}
        X\succsim_1 Y~(X\succ_1 Y) \implies X\succsim_2 Y~(X\succ_2 Y),~~\mbox{ for all  }X\in \mathcal{X}(\mathcal{G})~\text{and}~Y\in \mathcal{X}
    \end{equation} holds if and only if the $\succsim_1$-utility is equal to the $\succsim_2$-utility, the $\succsim_1$-distortion is equal to the $\succsim_2$-distortion, and $\succsim_2$ is more ambiguity averse than $\succsim_1$.
\end{proposition}

As  Proposition \ref{prop:attitudes}  demonstrates, our notion of comparative ambiguity attitudes for CRDU preference relations generalizes the comparative ambiguity attitudes of \cite{E04} given by \eqref{eq:ep}.

Our next goal is to define a normalization for ambiguity neutrality.  {Similar to the definition of comparative ambiguity attitudes, we define ambiguity neutrality as a property of the preference relation restricted to events.}

\begin{enumerate}[(SRM)]\setlength{\itemsep}{1pt}
    \item[(AN)] Ambiguity neutrality: For all $A,B,C\in \mathcal{F}$ satisfying $(A\cup B)\cap C=\varnothing$,
    \begin{equation} 
        A\succsim B\iff A\cup C\succsim A\cup C.
    \end{equation}
\end{enumerate}

 {We have the following theorem characterizing Property (AN). However, this characterization only requires the four regularity axioms, which, by Proposition \ref{proposition:match}, guarantee the existence and uniqueness of a matching probability.}

\begin{theorem}\label{th:AN}
    Let $\succsim$ satisfy Axioms (RC), (M), (SRM), and (C). Then $\succsim$ satisfies Property (AN) if and only if the $\succsim$-matching probability $\nu$ is a probability measure.
\end{theorem}

In \cite{E04}, the normalization for ambiguity neutrality was given by preference relations consistent with respect to first-order stochastic dominance for some $\mu\in\mathcal{M}_1$. That is, there exists $\mu\in \mathcal{M}_1$ such that 
\begin{equation}
    \label{eq:FSD}X\geq_{\mu}^{\mathrm{sd}}Y\implies X\succsim Y,~~X,Y\in \mathcal{X}.
\end{equation}
If $\succsim$ satisfies \eqref{eq:FSD}, then $\succsim$ is probabilistically sophisticated with respect to $\mu$, i.e., for all $X,Y\in \mathcal{X}$ satisfying $X=_{\mu}^{\mathrm{sd}}Y$, we have $X\simeq Y$; see \cite{MS92}. The following proposition demonstrates that our normalization for ambiguity neutrality aligns with the one presented in \cite{E04}.

\begin{proposition}\label{prop:FSD}
    Let $\succsim$ be a CRDU preference relation. Then $\succsim$ satisfies Property (AN) if and only if there exists $\mu\in \mathcal{M}_1$ such that \eqref{eq:FSD} holds.
\end{proposition}

Consistency with respect to first-order stochastic dominance, i.e., \eqref{eq:FSD}, does not represent ambiguity neutrality in the general Savage setting. For instance, in the MEU model of \cite{GS89}, there are preference relations that are consistent with respect to first-order stochastic dominance, but do not have a prior set given by a singleton.
Consider, for example, when $Q_{\epsilon}=\{\mu\in \mathcal{M}_1:H(\mu|\mathbb{Q})\leq\epsilon\},$ where $\mathbb{Q}\in \mathcal{M}_1$ is a reference probability measure and $\epsilon>0$. The preference relation $\succsim$ given by
$$X\succsim Y\iff \min_{\mu \in \mathcal{Q}_{\epsilon}}\int_{\Omega} u(X)\d\mu\geq \min_{\mu \in \mathcal{Q}_{\epsilon}}\int_{\Omega} u(Y)\d\mu,~~X,Y\in \mathcal{X}$$
 satisfies \eqref{eq:FSD} with $\mu=\mathbb{Q}$; see \citet[Corollary 1]{S13}. As sets of multiple priors in the MEU model represent both the perception of ambiguity and a strong form of ambiguity aversion, as discussed by \cite{GMM04}, consistency with respect to first-order stochastic dominance is not sufficient to claim ambiguity neutrality.
 Example \ref{ex:HS} gives a similar example in the framework of variational preferences.
 
In our Savage setting with pure risk, Axiom (RC) enables us to pin down the DM's risk preferences. Furthermore, by Theorem \ref{th:AN}, Property (AN) implies that the DM transforms reduces all acts to lotteries {via the matching probability in $\mathcal{M}_1$}, and then uses their risk preferences on these lotteries to make decisions. This is why the example for the MEU model presented above does not capture ambiguity neutrality. In the MEU model, the risk preferences are given by the EU model. Under the situation where the set of priors is not a singleton, but the preferences are consistent with respect to first-order stochastic dominance, the preference relation induced on lotteries is not EU. In the example above, the induced preference relation on lotteries coincides with the so-called coherent entropic risk measure (\citealp{FK11} and \citealp{AJ12}) composed with the utility transformation. Therefore, ambiguity neutrality, in general, should be equivalent to consistency with respect to first-order stochastic dominance and the induced preference relation on lotteries coinciding with the risk preferences. This definition is consistent with \cite{GM02}, where the authors' normalization for ambiguity neutrality corresponds to preference relations from the SEU model.


Finally, we say that a CRDU preference relation $\succsim$ is \emph{ambiguity averse (AA)} if it is more ambiguity averse than a CRDU preference relation satisfying Property (AN). A straightforward consequence of Theorem \ref{prop:comAm} and Theorem \ref{th:AN} is that given a CRDU preference relation $\succsim$, Property (AA) is equivalent to the $\succsim$-matching probability being balanced. {The idea of using a balanced matching probability in the CEU model to capture a form of ambiguity aversion was discussed in \cite{CT02}. For alternative forms of ambiguity aversion in the CEU model and their connection to matching probabilities, see \cite{HK25}.}

\subsection{Propensity for diversification}

\label{sec:uncertaintyaversion}

The idea of propensity for diversification in finance has appeared in the setting of risk since the foundational work of \cite{M52}. For decision models that incorporate ambiguity, such a property has been formalized in the classic papers by \cite{GS89} and \cite{S89}.\footnote{In \cite{GS89} and \cite{S89}, diversification seeking is called uncertainty aversion.} We present the definition of this property below, which is also sometimes called the convexity of the preferences.
\begin{enumerate}[(SRM)]\setlength{\itemsep}{1pt}
     \item[(DS)] Diversification seeking: For all $X,Y\in \mathcal{X}$ satisfying $X\simeq Y$ and $\lambda\in [0,1]$,
    $$\lambda X+(1-\lambda)Y\succsim X.$$
\end{enumerate}
If the preference relation of the DM satisfies Property (DS), then the DM prefers mixtures of equally preferred acts. This aligns with the intuitive idea of diversification. We immediately get the following proposition based on the work of \cite{WY19}.

\begin{proposition}
    \label{prop:UA}
    Let $\succsim$ be a CRDU preference relation in \eqref{eq:rep}. Then $\succsim$ satisfies Property (DS) if and only if the  $u$ is concave and the capacity $g\circ \nu$ is supermodular.
\end{proposition}

By Proposition \ref{prop:UA}, for a CRDU preference relation $\succsim$, Property (DS) depends on both the risk component of $\succsim$, i.e., the $\succsim$-utility and the $\succsim$-distortion, and the ambiguity component of $\succsim$, i.e., the $\succsim$-matching probability. Furthermore, since the $\succsim$-matching probability $\nu$ is risk conforming and $(\Omega,\mathcal{G},\mathbb{P})$ is atomless, it follows from the supermodularity of $g\circ\nu$, where $g$ is the $\succsim$-distortion, that $g$ is convex, see \citet[Proposition 4.75]{FS16}. Consequently, the Property (DS) implies a well-known risk attitude, which we now recall.

Given $\mu\in \mathcal{M}_1$ and the acts $X,Y\in \mathcal{X}$, we say that
\emph{$X$ dominates $Y$ in second-order stochastic dominance (SSD) under $\mu$}, written $X\geq^{\mathrm{ssd}}_{\mu}Y$,  if for all increasing concave functions $f:\mathbb{R}\to\mathbb{R}$, it holds that
$\int_{\Omega}f(X)\d\mu\geq \int_{\Omega}f(Y)\d\mu.$ SSD is closely related to strong risk aversion as studied by \cite{RS70}. Note that the following property is defined only for pure-risk acts, that is, acts in $\mathcal{X}(\mathcal{G})$. 
\begin{enumerate}[(SRM)]\setlength{\itemsep}{1pt}
    \item[(SRA)] Strong risk aversion:
    For all $X,Y\in \mathcal{X}(\mathcal{G})$ satisfying $X\geq^{\mathrm{ssd}}_{\mathbb{P}} Y$, $X\succsim Y$.
\end{enumerate}
A direct consequence of \cite{SZ08} is that a CRDU preference relation $\succsim$ satisfies Property (SRA) if and only if the $\succsim$-utility is concave and the $\succsim$-distortion is convex.  Therefore, Property (DS) implies Property (SRA). 

 Property (SRA) is strictly weaker than Property (DS), as Property (DS) is a consequence of both the risk and ambiguity components of the preferences, but Property (SRA) is only a consequence of the risk component of the preferences. Ideally, one would like to identify an ambiguity attitude that, together with Property (SRA), implies Property (DS). The following example shows that this can be done for the specific class of CEU preference relations.
\begin{example}
\label{ex:SAA}
    Let $\succsim$ be a CEU preference relation; that is,
    $$X\succsim Y\iff \int_{\Omega}u(X)\d\nu\geq \int_{\Omega}u(Y)\d\nu,~~X,Y\in \mathcal{X},$$
    where $u$ is the $\succsim$-utility and $\nu$ is the $\succsim$-matching probability. By Proposition \ref{prop:UA}, we know that $\succsim$ satisfies Property (DS) if and only if $u$ is concave and $\nu$ is supermodular. Since the supermodularity of $\nu$ implies that $\nu$ is balanced, we have $\succsim$ satisfies Property (AA). Therefore, the supermodularity of the $\succsim$-matching probability is a stronger form of ambiguity aversion. Furthermore, this stronger form of ambiguity aversion, combined with Property (SRA)  guaranteeing the concavity of $u$,  yields Property (DS).
\end{example}

For general CRDU preferences, the forward direction of the equivalence established in Example \ref{ex:SAA} for CEU preferences still holds, as shown in the following theorem.

\begin{theorem}
\label{prop:main}
    Let $\succsim$ be a CRDU preference relation. If $\succsim$ satisfies Property (SRA) and the $\succsim$-matching probability is supermodular, then $\succsim$ satisfies Property (DS).
\end{theorem}

Therefore, a combination of both strong risk aversion (Property (SRA)) and a stronger version of ambiguity aversion (Property (AA)) will ensure that the DM prefers diversification (Property (DS)). However, the next example demonstrates that the reverse direction of the equivalence established in Example \ref{ex:SAA} for CEU preferences does not generally extend to CRDU preferences. 
In particular, we show the existence of a CRDU preference relation $\succsim$ that satisfies Property (DS) but not Property (AA), which implies that the $\succsim$-matching probability is not supermodular. 

\begin{example}
    \label{ex:counter}
 {Assume there exists a risk-conforming $\mathbb{Q}\in \mathcal{M}_1$} and a sub-$\sigma$-algebra $\mathcal{H}$, independent of $\mathcal{G}$ under $\mathbb{Q}$,\footnote{Independence between $\mathcal{H}$ and $\mathcal{G}$ means $\mathbb{Q}(A\cap B)=\mathbb{Q}(A)\mathbb{Q}(B)$  for all $A\in \mathcal{G}$ and $B\in \mathcal{H}$.} such that there exists $A_0\in \mathcal{H}$ satisfying $\mathbb{Q}(A_0)\in (0,1)$. Given a strictly concave continuous distortion function $h$, we show in Appendix \ref{App:prop} the existence of a continuous risk-conforming capacity $\nu$ for which
    $$h(\mathbb{Q}(A))\geq\nu(A)\geq \mathbb{Q}(A)~\text{for all}~A\in \mathcal{F}, ~~~~\nu(A)=h(\mathbb{Q}(A))~\text{for all}~A\in \mathcal{H},$$
    and $g\circ\nu$ is supermodular, where $g$ is the inverse distortion function of $h$. Fix the strictly increasing concave vNM utility $u$ and define the CRDU preference relation $\succsim$ by
    $$X\succsim Y\iff \int_{\Omega}u(X)\d(g\circ\nu)\geq \int_{\Omega}u(Y)\d(g\circ\nu),~~X,Y\in \mathcal{X}.$$
    As $g\circ\nu$ is supermodular and $u$ is concave, by Proposition \ref{prop:UA}, $\succsim$ satisfies Property (DS). By the strict concavity of $h$,
    $\nu(A_0)=h(\mathbb{Q}(A_0))>\mathbb{Q}(A_0)h(1)+(1-\mathbb{Q}(A_0))h(0)=\mathbb{Q}(A_0).$
    Similarly, as $\mathbb{Q}(A_0^c)\in (0,1)$, $\nu(A_0^c)>\mathbb{Q}(A_0^c)$. Thus, $\nu(A)+\nu(A^c)>1$ and $\succsim$ does not satisfy Property (AA).
\end{example}

The preference relation $\succsim$ constructed in Example \ref{ex:counter} further illustrates the interplay between the risk and ambiguity components of CRDU preferences. Suppose that we fix the generalized probabilistic beliefs of $\succsim$ from Example \ref{ex:counter} (we fix the $\succsim$-matching probability $\nu$), but replace the risk preferences with the expected-value risk preferences, namely,
$$X\succsim' Y\iff \int_{\Omega}X\d\nu\geq \int_{\Omega}Y\d\nu,~~X,Y\in \mathcal{X}.$$
Therefore, it holds that
    $$X\succsim'Y\iff \int_{\Omega}X\d(h\circ\mathbb{Q})\geq \int_{\Omega}Y\d(h\circ\mathbb{Q}),~~X,Y\in \mathcal{X}(\mathcal{H}).$$
Since $h\circ \mathbb{P}$ is submodular (\citealp[Proposition 4.75]{FS16}), for all $X,Y\in \mathcal{X}(\mathcal{H})$, if $X\simeq' Y$, then for all $\lambda\in (0,1)$, $X\succsim'\lambda X+(1-\lambda)Y$. Therefore, for the preference relation $\succsim'$, there is an aversion to diversification locally on $\mathcal X(\mathcal H)$, contrasting the global propensity for diversification of $\succsim$, although both preference relations satisfy Property (SRA). 


\subsection{The robust interpretation}

\label{sec:robust}

A preference relation $\succsim$ is called a \emph{maxmin expected utility (MEU)} preference relation (\citealp{GS89}) if there exists a collection of risk-conforming probability measures $\mathcal{S}$\footnote{Here, we suppose each $\mu\in \mathcal{S}$ is risk conforming to ensure we remain in the Savage setting with pure risk. That is, since $\mathbb{P}$ is well specified on $\mathcal{G}$, it should hold that $\mu|_{\mathcal{G}}=\mathbb{P}$ for all $\mu\in\mathcal{S}$.}   and a vNM utility $u$ such that
\begin{equation}
    X\succsim Y\iff \min_{\mu\in \mathcal{S}}\int_{\Omega}u(X)\d\mu\geq \min_{\mu\in \mathcal{S}}\int_{\Omega}u(Y)\d\mu,~~X,Y\in \mathcal{X}.
\end{equation}
Intuitively, a DM with an MEU preference relation chooses the set $\mathcal{S}$ to represent their ambiguity perception and then aggregates their risk preferences, given by the EU model, over this set. This form of aggregation is known as the robust aggregation of risk preferences, and it corresponds to a very strong form of ambiguity aversion; see \cite{GMM04}. 

For the CRDU preference relation $\succsim$ in \eqref{eq:rep}, the above idea of robust aggregation of risk preferences means the following property: 
\begin{equation}
        \label{eq:MM}
        X\succsim Y\iff \min_{\mu\in \mathfrak{C}(\nu)}\int_{\Omega}u(X)\d(g\circ\mu)\geq \min_{\mu\in \mathfrak{C}(\nu)}\int_{\Omega}u(Y)\d(g\circ\mu),~~X,Y\in \mathcal{X},
    \end{equation} 
where the set of probability measures $\mathfrak{C}(\nu)$ is understood as reflecting the DM’s ambiguity perception. As the following theorem demonstrates, the supermodularity of the $\succsim$-matching probability enables the robust aggregation of risk preferences in the case of CRDU preference relations.

\begin{theorem}
    \label{prop:maxmin}
    Let $\succsim$ be a CRDU preference relation in \eqref{eq:rep}. Then the robust aggregation \eqref{eq:MM} holds if and only if $\nu$ is supermodular.
\end{theorem}

It is well-known that Property (DS) connects CEU preference relations to MEU preference relations  (\citealp{S89}). In the case of CRDU preference relations, Property (DS) is not sufficient to guarantee the robust aggregation of risk preferences. To see this, let $\succsim$ be a CRDU preference relation in \eqref{eq:rep} satisfying Property (DS). Then, by Proposition \ref{prop:UA} and \citet[Theorem 4.6]{MM04}, it holds that
\begin{equation}
    \label{eq:OGMM}
    X\succsim Y\iff \min_{\mu\in \mathfrak{C}(g\circ\nu)}\int_{\Omega}u(X)\d\mu\geq \min_{\mu\in \mathfrak{C}(g\circ \nu)}\int_{\Omega}u(Y)\d \mu,~~X,Y\in \mathcal{X}.
\end{equation}
The expression in \eqref{eq:OGMM} represents an aggregation, but it does not correspond to the robust aggregation of the DM's risk preferences. This follows since \eqref{eq:OGMM} corresponds to an aggregation of the EU risk preferences. However, as $\succsim$ is a CRDU preference relation, the risk preferences should be given by the RDU model with distortion $g$. Furthermore, the set $\mathfrak{C}(g\circ \nu)$ in \eqref{eq:OGMM} should not be interpreted as the DM's ambiguity set as the capacity $g\circ \nu$ does not represent the ambiguity component of $\succsim$; rather, the $\succsim$-matching probability $\nu$ represents the ambiguity component of $\succsim$. 

\section{An exogenous global probability measure}

\label{sec:alt}

{In this section, we will assume that the DM has, through some statistical procedure, estimated a reference probability measure $\mathbb{Q}\in \mathcal{M}_1$ to initially describe their uncertainty. We will assume that the reference probability measure agrees with the likelihood evaluations on the unambiguous events; that is, $\mathbb{Q}$ is risk conforming. There are two ways to interpret the relationship between $\mathbb{Q}$ and $\mathbb{P}$. The first is that the DM begins with the exogenous unambiguous events and the probability measure $\mathbb{P}$. From here, when estimating the model $\mathbb{Q}$, the DM ensures that the statistical procedure produces a risk-conforming $\mathbb{Q}$. Example \ref{ex:product} below illustrates this interpretation in a setup used by \cite{KMM05}.}  {For the second way of interpreting this assumption,  the DM can first estimate $\mathbb{Q}$ and then, through introspection or some quantitative statistical procedure, designate $\mathcal{G}$ as being unambiguous. Then, the DM takes $\mathbb{P}=\mathbb{Q}|_{\mathcal{G}}$. This modeling approach was discussed in \cite{SVW24}. 

Regardless of the interpretation, the key motivating takeaway is that the DM has a global probability measure $\mathbb{Q}$, but a priori, only fully trusts the model on $\mathcal{G}$.}

\begin{example}
\label{ex:product}  Let $(\Omega_1,\mathcal{F}_1)$ denote a measurable space where the DM lacks any prior probabilistic information. Let $(\Omega_2,\mathcal{F}_2,\mathbb{P}_2)$ denote a probability space where the DM knows that observations from this space happen according to $\mathbb{P}_2$. For example, the probability space $(\Omega_2,\mathcal{F}_2,\mathbb{P}_2)$ could represent a lottery generating device. Define $(\Omega,\mathcal{F})=(\Omega_1\times \Omega_2,\mathcal{F}_1\otimes\mathcal{F}_2)$. Define $\mathcal{G}=\{\Omega_1\times A:A\in \mathcal{F}_2 \}$ and $\mathbb{P}:\mathcal{G}\to[0,1]$ by $$\mathbb{P}(\Omega_1\times A)=\mathbb{P}_2(A),~~A\in \mathcal{F}_2.$$
    Assume that the DM, through some statistical estimation procedure, has estimated the probability measure $\mathbb{P}_1$ for $(\Omega_1,\mathcal{F}_1)$. Additionally, assume the DM believes that these two sources of randomness are independent of each other, which is consistent with the interpretation that $(\Omega_2,\mathcal{F}_2,\mathbb{P}_2)$ represents a lottery generating device. Therefore, the DM takes as their reference probability measure $\mathbb{Q}\in \mathcal{M}_1$ the unique probability measure satisfying
    $$\mathbb{Q}(A_1\times A_2)=\mathbb{P}_1(A_1)\mathbb{P}_2(A_2),~~A_1\in \mathcal{F}_1~\text{and}~A_2\in \mathcal{F}_2.$$
    It is straightforward to verify that $\mathbb{Q}$ is risk conforming.
\end{example}

Consider the case where the DM estimated the model $\mathbb{Q}$ to inform investment decisions for their financial portfolio, which they plan to liquidate within a week. Under the reference model $\mathbb{Q}$, assume that the weekly losses of a well-established equity $X$ and the weekly losses of a recently launched cryptocurrency $Y$ share the same distribution. Because the equity has been observed for an extended period, the DM perceives the events related to $X$ as unambiguous, i.e., $X\in \mathcal{X}(\mathcal{G})$. When deciding between an investment in $X$ versus an investment in $Y$, the DM is automatically drawn to $X$ because it has the same estimated distribution as $Y$, without any ambiguity. Therefore, the assumption $X\succsim Y$ is reasonable. This discussion motivates the following property.
\begin{enumerate}[(SRM)]\setlength{\itemsep}{1pt}
    \item[(RAA)] Reference ambiguity aversion: For all $X\in \mathcal{X}(\mathcal{G})$ and $Y\in \mathcal{X}$ satisfying $X=_{\mathbb{Q}}^{\mathrm{sd}}Y$, $X\succsim Y$.
\end{enumerate}
{As $\mathbb{Q}$ is risk conforming,} it is straightforward to verify that Property (RAA) is a stronger condition than Axiom (RC). The terminology ``reference ambiguity aversion" comes from the following proposition.

\begin{proposition}
    \label{prop:RAA}
    Let $\succsim$ be a CRDU preference. Then $\succsim$ satisfies Property (RAA) if and only if the $\succsim$-matching probability $\nu$ is set-wise dominated by the reference probability measure $\mathbb{Q}$, i.e., $\mathbb{Q}\in \mathfrak{C}(\nu)$.
\end{proposition}
If the CRDU preference relation $\succsim$ satisfies Property (RAA), then, by Proposition \ref{prop:RAA}, the $\succsim$-matching probability must be balanced.  In general, we have 
$$
\mbox{Property (RAA)} \implies \mbox{Property (AA) and Axiom (RC)} 
$$
To see that the converse implication does not hold,  we can take $\nu$ such that $\mathfrak C (\nu)$ does not contain $\mathbb{Q}$ but contains some risk conforming $\mu\in \mathcal{M}_1$ that differs from $\mathbb{Q}$ outside $\mathcal G$.

We now derive an alternative representation for CRDU preference relations when the following additional property is satisfied.
\begin{enumerate}[(SRM)]\setlength{\itemsep}{1pt}
    \item[(NSC)] Null set consistency: For all $X,Y\in \mathcal{X}$ satisfying $X=^{\mathrm{as}}_{\mathbb{Q}}Y$, $X\simeq Y$.
\end{enumerate}
Property (NSC) encodes the idea that the DM is confident that the reference probability measure $\mathbb{Q}$ correctly identifies sets of probability zero and is indifferent to acts that only differ on a set of probability zero.  Property (NSC) is not implied by Axiom (RC), which does not restrict $\succsim$ outside $\mathcal G$. For $A,B\in \mathcal{F}$, denote by $A\triangle B$ the symmetric set difference, that is, $A\triangle B=(A\setminus B)\cup (B\setminus A)$. For a capacity $\nu$, we say that $\nu$ is \emph{$\mathbb{Q}$-consistent} if for all $A,B\in \mathcal{F}$ with $\mathbb{Q}(A\triangle B)=0$, $\nu(A)=\nu(B)$. 
\begin{proposition}
    \label{prop:nullity}
    Let $\succsim$ be a CRDU preference relation. Then $\succsim$ satisfies Property (NSC) if and only if the $\succsim$-matching probability is $\mathbb{Q}$-consistent.
\end{proposition} 

We now provide an alternative representation of CRDU when Property (NSC) is also assumed. For this purpose, we need the following definitions of distortion families.  In the sequel, we write $\mathcal{D}(X)=\{\mathbb{Q}(X>x):x\in \mathbb{R}\}\subseteq [0,1]$.

\begin{definition} 
\begin{enumerate}
    \item[(i)]
An indexed family $(g_X)_{X\in \mathcal{X}}$ is called a \emph{distortion family} if $g_X$ is distortion function for all $X\in \mathcal{X}$ and the following three properties are satisfied: (a) For all $A,B\in \mathcal{F}$ satisfying $A\subseteq B$, $g_{\id_A}(\alpha)\leq g_{\id_B}(\alpha)$ for all $\alpha \in [0,1];$  (b) For all $X\in \mathcal{X}$, $g_X$ is continuous when restricted to $\mathcal{D}(X)$; (c) For all comonotonic $X,Y\in \mathcal{X}$, $g_X(\alpha)=g_Y(\alpha)$ for all $\alpha\in \mathcal{D}(X)\cap \mathcal{D}(Y).$ 
    \item[(ii)] For  a distortion function  $g$, 
    a distortion family $(g_X)_{X\in \mathcal{X}}$ is called a \emph{$g$-distortion family} if for all $X\in \mathcal{X}(\mathcal{G})$, $g_X(\alpha)=g(\alpha)$ for all $\alpha\in \mathcal{D}(X),$
 A $g$-distortion family $(g_X)_{X\in \mathcal{X}}$ is \emph{ambiguity averse} if for all $X\in \mathcal{X}$, $g_X(\alpha)\leq g(\alpha)$ for all $\alpha\in \mathcal{D}(X)$. 
\end{enumerate}
\end{definition}

\begin{theorem}
    \label{theo:2}
The CRDU preference $\succsim$ in \eqref{eq:rep} with Property (NSC) has the alternative representation 
    \begin{equation}
    \label{eq:secondRep}
        X\succsim Y \iff \int_{\Omega} u(X)\d(g_X\circ\mathbb{Q})\geq \int_{\Omega} u(Y)\d(g_Y\circ\mathbb{Q}),~~X,Y\in \mathcal{X}.
    \end{equation}
    where $(g_X)_{X\in \mathcal{X}}$ is a
     $g$-distortion family that is uniquely determined on $\mathcal D(X)$ for each $X$. Furthermore, $\succsim$ satisfies (RAA) if and only if $(g_X)_{X\in \mathcal{X}}$ is ambiguity averse.  
\end{theorem}

The representation in Theorem \ref{theo:2} gives an alternative view on how the DM incorporates the ambiguity component of their preferences. Intuitively, the DM uses the reference measure $\mathbb{Q}$ to make decisions; however, to account for ambiguity, how the reference measure is distorted depends on the act being evaluated. Since there is no distributional ambiguity for the acts $\mathcal{X}(\mathcal{G})$, the distortion functions associated with acts in $\mathcal{X}(\mathcal{G})$ coincide. The common distortion function fixes how the DM distorts unambiguous probabilities and is an aspect of the risk component of the preferences. Given some $X\in \mathcal{X}\setminus \mathcal{X}(\mathcal{G})$,  how $g_X$ differs from $g$ encodes the DM's perception of and attitude toward the ambiguity in $X$.

\begin{remark}
    We show in Appendix \ref{app:alt} that representation \eqref{eq:secondRep} is indeed a sufficient condition for $\succsim$ to be a CRDU preference relation.  The representation \eqref{eq:secondRep} is primarily intended to motivate an alternative view on how the DM incorporates the ambiguity component of their preferences, rather than as a reference point for modeling such preference relations.
\end{remark}

Assume that $\succsim$ satisfies representation \eqref{eq:secondRep}. Given some $X\in \mathcal{X}\setminus \mathcal{X}(\mathcal{G})$, if it holds that $g_X=g$, then the DM perceives no ambiguity in $X$ or is neutral towards the ambiguity perceived in $X$. For example, if $g_X=g$ holds for every $X\in \mathcal{X}$, then $\succsim$   satisfies  
$$X\succsim Y\iff \int_{\Omega}u(X)\d(g\circ \mathbb{Q})\geq \int_{\Omega}u(Y)\d(g\circ \mathbb{Q}),~~X,Y\in \mathcal{X}.$$
That is, the DM uses the RDU model globally under the reference probability measure $\mathbb{Q}$ and, thereby, satisfies Property (AN). If $\succsim$ satisfies Property (RAA), then, by Theorem \ref{theo:2}, for all $X\in \mathcal{X}$, $$g_X(\alpha)\leq g(\alpha),~~\alpha \in \mathcal{D}(X)$$ Therefore, for all acts $X\in\mathcal{X}\setminus \mathcal{X}(\mathcal{G})$, the distortion $g_X$ will be more risk averse than the DM's risk distortion $g$, in the classic sense of \cite{P64} and \cite{Y69}.

\section{Risk preferences given by the dual utility model}

\label{sec:dual}

The dual utility model of \cite{Y87} is the special case of the RDU model when the vNM utility is linear. The goal of this section is to axiomatize the subclass of CRDU preferences that have risk preferences given by the dual utility model. For this purpose, we introduce the following axiom, originally formulated by \cite{Y87} in the risk setting and by \cite{S89} in the setting of ambiguity.
\begin{enumerate}[(SRM)]\setlength{\itemsep}{1pt}
    \item[(CI)] Comonotonic independence: For all pairwise comonotonic $X,Y,Z\in \mathcal{X}$, 
    $$X\simeq Y\Rightarrow(1/2) X + (1/2) Z\simeq (1/2) Y + (1/2)Z.$$
\end{enumerate}
Intuitively, the Axiom (CI) says that an act, mutually comonotonic to two indifferent comonotonic acts, cannot be used as a way to hedge, i.e., as a way to change the indifference. Axiom (CI) is structurally similar to Axiom (SCI) presented in Section \ref{sec:ax1}. However, Axiom (CI) is much simpler to define as it does not need the subjective mixture operation for acts. It turns out that Axiom (CI), together with the earlier regularity axioms, suffices to characterize all CRDU preference relations whose risk preferences are given by the dual utility model.

\begin{theorem}
    \label{theo:dual}
    The preference relation $\succsim$ satisfies Axioms (RC), (M), (SRM), (C), and (CI) if and only if 
    there exists a strictly increasing continuous distortion function $g$ and a continuous risk-conforming capacity $\nu$ such that
         \begin{equation}
        X\succsim Y \iff \int_{\Omega} X\d(g\circ\nu)\geq \int_{\Omega} Y\d(g\circ\nu),~~X,Y\in \mathcal{X}. \label{eq:repD}
    \end{equation} 
    Moreover, in the representation \eqref{eq:repD}, $g$ is unique and $\nu$ is unique and coincides with the $\succsim$-matching probability.
\end{theorem}

Any preference relation satisfying representation \eqref{eq:repD} is evidently a CRDU preference relation. This consequence is further clarified by the following observation. If a preference relation $\succsim$ satisfies the Axioms (RC), (M), (SRM), and (C), then, given $x,y\in \mathbb{R},$ we can define the subjective mixture of these outcomes $(1/2)x\oplus(1/2)y$; see Section \ref{sec:ax1}. If we further assume that $\succsim$ satisfies Axiom (CI), then it holds that $(1/2)x\oplus (1/2)y=(1/2)x+(1/2)y$, that is, the subjective mixture is given by the usual arithmetic average. Therefore, Axiom (CI) implies Axiom (SCI).

\section{Conclusion}

We provided a new axiomatization of the Choquet rank-dependent utility model, an ambiguity model that is structurally similar to the Choquet expected utility model, but when confined to pure risks, it reduces to the rank-dependent utility model, rather than to the expected utility model. Our model is formulated in the Savage setting, instead of the Anscombe--Aumann setting.

The main result, Theorem \ref{theo:good}, axiomatizes this model in the Savage setting with pure risk and, notably, requires only two axioms apart from some standard regularity axioms. The ambiguity component of the preferences is fully captured by the notion of a matching probability. To study ambiguity attitudes, we define comparative ambiguity attitudes, which involve only preference comparisons between binary acts rather than all acts. Specifically, Theorem \ref{prop:comAm} shows that, for CRDU preference relations, comparative ambiguity attitudes are fully characterized by their matching probabilities. After fixing a normalization for ambiguity neutrality, characterized in Theorem \ref{th:AN}, we define and characterize absolute ambiguity attitudes. 
On the interaction between risk and ambiguity attitudes shapes the DM’s propensity for diversification, 
Theorem \ref{prop:main} shows that strong risk aversion and the supermodularity of the matching probability are sufficient for propensity for diversification. However, unlike the Choquet expected utility model, these conditions are not necessary. This result allows a model to maintain the propensity for diversification by adopting flexible risk and ambiguity attitudes. In addition, Theorem \ref{prop:maxmin} shows that the supermodularity of the matching probability is equivalent to the preference relation admitting a robust aggregation interpretation of the risk preferences. To give a complementary perspective on the ambiguity component of CRDU preference relations, Theorem \ref{theo:2} provides a representation derived under an additional property. This representation states that the DM distorts a reference probability measure, but allows this distortion to depend on the act being evaluated. The extent to which the distortion function associated with an act differs from the risk distortion determines the DM’s perception of and attitude toward the ambiguity in that act. Finally, Theorem \ref{theo:dual} provides an axiomatization for the subclass of CRDU preference relations whose risk preferences are given by the dual utility model.  This axiomatization is concise as, apart from the regularity axioms from Theorem \ref{theo:good}, only the single axiom of comonotonic independence is needed.

\subsection*{Acknowledgments}
The authors thank Jean-Gabriel Lauzier for many useful discussions, and Peter Wakker for helpful comments on an early version of the paper and especially for bringing the paper of \cite{B25} to our attention. Zachary Van Oosten is supported by an OGS/QEII Graduate Scholarship. 
Ruodu Wang is supported by the Natural Sciences and Engineering Research Council of Canada (CRC-2022-00141, RGPIN-2024-03728).

\label{sec:con}

\newpage

\appendix

\begin{center}
    \Large Appendices: Omitted proofs and additional results
\end{center}

Appendices \ref{ap:ax1}-\ref{app:dual} present the proofs omitted from the main paper and some additional results. Appendix \ref{ap:mod} provides a general result for lattices that is used in the proof of \ref{prop:main}. 

\section{Results and proofs accompanying Section \ref{sec:ax1}}
\label{ap:ax1}

\begin{proof}[Proof of Proposition \ref{prop:consis}]
    Since $\succsim$ is a CRDU preference relation, there exists a continuous strictly increasing vNM utility $u$, a continuous strictly increasing distortion function $g$, and a continuous risk-conforming capacity $\nu$ such that \eqref{eq:rep} holds true. Since $\nu$ is risk conforming, it is clear that $\succsim$ satisfies Axiom (RC). The fact that $\succsim$ satisfies Axiom (M) follows from the monotonicity of the Choquet integral and the fact that $u$, by virtue of being a vNM utility, is increasing. 
   
   Let $X,Y\in \mathcal{X}(\mathcal{G})$ satisfy $X>_{\mathbb{P}}^{\mathrm{as}} Y$. Since $u$ is strictly increasing, $u(X)>_{\mathbb{P}}^{\mathrm{as}} u(Y)$. Therefore, for all $x\in \mathbb{R}$,
   $\mathbb{P}(u(X)>x)\geq \mathbb{P}(u(Y)>x)$, and there exists $x_0\in \mathbb{R}$ such that $\mathbb{P}(u(X)>x_0)>\mathbb{P}(u(Y)>x_0).$ Since $g$ is strictly increasing, for all $x\in \mathbb{R}$, $g(\mathbb{P}(u(X)>x))\geq g(\mathbb{P}(u(Y)>x))$, and  $g(\mathbb{P}(u(X)>x_0))>g(\mathbb{P}(u(Y)>x_0))$. By the right continuity of the mappings $$x\mapsto g(\mathbb{P}(u(X)>x))~~\text{and}~~x\mapsto g(\mathbb{P}(u(Y)>x)),$$
   there exists $\epsilon>0$ such that for all $x\in [x,x+\epsilon)$, $g(\mathbb{P}(u(X)>x))> g(\mathbb{P}(u(Y)>x))$. Therefore, $$\int_{\Omega}u(X)\d(g\circ \nu)=\int_{\Omega}u(X)\d(g\circ \mathbb{P})>\int_{\Omega}u(Y)\d(g\circ \mathbb{P})=\int_{\Omega}u(Y)\d(g\circ \nu),$$
   $X\succ Y$, and $\succsim$ satisfies Axiom (SRM).
   
   Let $Y\in \mathcal{X}$ and $(X_n)_{n\in N}\subseteq \mathcal{X}$ be a bounded sequence converging pointwise to $X\in \mathcal{X}$ satisfying $Y\succsim X_n$ for all $n\in \N$. Since $u$ is continuous and increasing, $(u(X_n))_{n\in \N}$ is bounded and converges pointwise to $u(X)$. As the capacity $g\circ \nu$ is continuous, by \citet[Theorem 11.9]{WK10A},  we have
   $$\int_{\Omega}u(X)\d(g\circ\nu)= \lim_{n\to\infty}\int_{\Omega}u(X_n)\d(g\circ\nu)\leq \int_{\Omega}u(Y)\d(g\circ\nu).$$
   Therefore, $Y\succsim X$. If we instead assumed that $X_n\succsim Y$ for all $n\in \N$, a similar argument would show that $X\succsim Y$. Therefore, $\succsim$ satisfies Axiom (C). 
\end{proof}

\begin{proof}[Proof of Proposition \ref{proposition:match}]
    Given $A\in \mathcal{F}$, define $\mathcal{L}(A)=\{\mathbb{P}(R):A\succsim R\text{ and }R\in \mathcal{G}\}$ and $\mathcal{U}(A)=\{\mathbb{P}(R):R\succsim A\text{ and }R\in \mathcal{G}\}.$
    Since, by Axiom (M),  $\Omega\succsim A\succsim\varnothing$, it holds that
    $0\in \mathcal{L}(A)\subseteq [0,1]$ and $1\in \mathcal{U}(A)\subseteq [0,1].$ Thus, both $
    \mathcal{L}(A)$ and $\mathcal{U}(A)$ are non-empty. Let $a_A=\sup(\mathcal{L}(A))$ and $b_A=\inf(\mathcal{U}(A))$, we claim that $a_A\in \mathcal{L}(A)$. To see this, let $(R_n)_{n\in \N} \subseteq \mathcal{G}$ with $A\succsim R_n$ such that $\mathbb{P}(R_n)\uparrow a_A$. If $\mathbb{P}(R_n)=a_A$ for some $n\in \N$, then $a_A\in \mathcal{L}(A)$. Therefore, we will assume that $\mathbb{P}(R_n)<a_A$ for all $n\in \N$. Since $(\Omega,\mathcal{G},\mathbb{P})$ is atomless, there exists a $\mathcal{G}$-measurable $U:\Omega\to(0,1)$ with $\mathbb{P}(U\leq x)=x$ for all $x\in (0,1)$, i.e., $U$ has a uniform distribution under $\mathbb{P}$. Define $R'_n=\{U\leq \mathbb{P}(R_n)\}$. By Axiom (RC), $A\succsim R'_n$. Since $\id_{R'_n}\to \id_{\{U< a_A\}}$ pointwise, by Axiom (C), $A\succsim \{U< a_A\}$. Therefore $a_A = \mathbb{P}(U< a_A)\in \mathcal{L}(A)$. Similarly, one can show that $b_A\in \mathcal{U}(A)$. Therefore, as $(\Omega,\mathcal{G},\mathbb{P})$ is atomless, by Axiom (M), $\mathcal{L}(A)=[0,a_A]$ and $\mathcal{U}(A)=[b_A,1]$. Furthermore, as
    $\mathcal{U}(A)\cup \mathcal{L}(A)=[0,1]$, $a_A\geq b_A$. For the sake of contradiction, assume that $a_A>b_A$. We can find $R_1,R_2\in \mathcal{G}$ such that $R_1\subseteq R_2$ and $\mathbb{P}(R_1)=b_A$ and $\mathbb{P}(R_2)=a_A$. By Axiom (SRM), $R_2\succ R_1$. On the other hand, by Axiom (RC), $R_1\succsim A\succsim R_2$, a contradiction. Therefore $a_A=b_A$. Since $A\in \mathcal{F}$ was general, we can define the set function $\nu:\mathcal{F}\to[0,1]$ by $\nu(A)=a_A=b_A$. By Axiom (RC) and the definition of $\mathcal{L}(U)$ and $\mathcal{U}(A)$, given any $R\in \mathcal{G}$ with $\mathbb{P}(R)=\nu(A)$, $A\simeq R$. Therefore, $\nu$ is a $\succsim$-matching probability. The uniqueness of $\nu$ is a consequence of the fact that $a_A=b_A$ for all $A\in \mathcal{F}$. Furthermore, the uniqueness of $\nu$ forces $\nu|_{\mathcal{G}}=\mathbb{P}$.
    
    We will show that \eqref{eq:setrep} is true. Let $A,B\in \mathcal{F}$ satisfy $A\succsim B$. Since $\nu$ is a matching probability, there exists $R_A,R_B\in \mathcal{G}$ satisfying $A\simeq R_A$ and $B\simeq R_B$ such that $\nu(A)=\mathbb{P}(R_A)$ and $\nu(B)=\mathbb{P}(R_B)$. Therefore, $R_A\simeq A\succsim B\simeq R_B$. For the sake of contradiction, assume that $\mathbb{P}(R_B)>\mathbb{P}(R_A)$. Under this assumption, as $(\Omega,\mathcal{G},\mathbb{P})$ is atomless, we can find $R_A^*$ such that $R_A^*\subseteq R_B$ and $\mathbb{P}(R_A^*)=\mathbb{P}(R_A)$. By Axioms (RC) and (SRM), $R_B\succ R_A^*\simeq R_A$, a contradiction. Thus, $\mathbb{P}(R_A)\geq \mathbb{P}(R_B)$ and $\nu(A)\geq \nu(B)$. Let $A,B\in \mathcal{F}$ satisfy $\nu(A)\geq \nu(B)$. Since $\nu$ is a matching probability, there exists $R_A,R_B\in \mathcal{G}$ satisfying $A\simeq R_A$ and $B\simeq R_B$ such that $\nu(A)=\mathbb{P}(R_A)$ and $\nu(B)=\mathbb{P}(R_B)$. As $(\Omega,\mathcal{G},\mathbb{P})$ is atomless, we can $R_B^*$ such that $R_B^*\subseteq R_A$ and $\mathbb{P}(R_B^*)=\mathbb{P}(R_B)$. Therefore, by Axioms (RC) and (M), $A\simeq R_A\succsim R_B^*\simeq R_B\simeq B.$ Therefore, \eqref{eq:setrep} holds and, by Axiom (M), $\nu$ is a capacity.

    We will show that $\nu$ is upward continuous. Let $(A_n)_{n\in \N}\subseteq \mathcal{F}$ be an increasing sequence with $A=\bigcup_{n\in\N}A_n$. For each $n\in \N$, define $R_n=\{U\leq \nu(A_n)\}$. By Axiom (RC) and the fact that $\nu$ is a matching probability, $R_n\simeq A_n$ for all $n\in N$. Let $R=\bigcup_{n\in\N}R_n$. By Axiom (M), $R\succsim R_n\simeq A_n$. As $\id_{A_n}\to \id_{A}$ pointwise, by Axiom (C), $R\succsim A$. By Axiom (M), $A\succsim A_n\simeq R_n$. As $\id_{R_n}\to \id_{R}$ pointwise, by Axiom (C), $A\succsim R$. Thus $A\simeq R$. Therefore,
    $$\nu(A)=\mathbb{P}(R)=\lim_{n\to\infty}\mathbb{P}(R_n)=\lim_{n\to\infty}\nu(A_n).$$
    A similar argument will show that $\nu$ is downward continuous. 
\end{proof}

\begin{proof}[Proof of Proposition \ref{prop:matchCheck}]
Let $A\in \mathcal{F}$ and $R_A\in \mathcal{G}$ such that $\nu(A)=\mathbb{P}(R_A)$, which is guaranteed
 to exist since $(\Omega,\mathcal{G},\mathbb{P})$ is atomless. As
   \begin{align*}
       \int_{\Omega}u(\id_A)\d(g\circ\nu)&=u(1)g(\nu(A))+u(0)(1-g(\nu(A)))\\&=u(1)g(\mathbb{P}(R_A))+u(0)(1-g(\mathbb{P}(R_A)))=\int_{\Omega}u(\id_{R_A})\d(g\circ\nu),
   \end{align*}
   $A\simeq R_A$. Therefore, $\nu$ is the $\succsim$-matching probability.
\end{proof}

\begin{proof}[Proof of Proposition \ref{prop:cert}]
    Let $X\in \mathcal{X}$, find $a,b\in \mathbb{R}$ such $a\geq X\geq b$. Therefore, by Axiom (M), the sets $\mathcal{L}_X=\{c\in \mathbb{R}:X\succsim c \}$ and $\mathcal{U}_X=\{c\in \mathbb{R}:c\succsim X\}$ are non-empty. It must hold that $\mathcal{L}_X\subseteq (-\infty,a]$, because if $X\succsim c$ for $c\in \mathbb{R}$ with $c>a$, then by Axiom (SRM), it would follow that
    $X\succ a\succsim X$. Similarly, we have that $\mathcal{U}_X\subseteq [b,\infty).$ Let $a^*=\sup(\mathcal{L}_X)\in \mathbb{R}$ and $b^*=\inf(\mathcal{U}_X)\in \mathbb{R}$. By Axioms (M) and (C), it is clear that
    $\mathcal{L}_X=(-\infty,a^*]$ and $\mathcal{U}_X=[b^*, \infty).$ As $\mathcal{L}_X\cup \mathcal{U}_X=\mathbb{R}$, $b^*\leq a^*$. Take $c_X\in \mathcal{L}_X\cap \mathcal{U}_X$. Uniqueness of $c_X$ follows from Axiom (SRM).
    
    Let $(X_n)_{n\in \N}\subseteq \mathcal{X}$ be a bounded sequence converging pointwise to $X\in \mathcal{X}$. For $n\in \N$, define $Z_n=\sup_{m\geq n}X_m$. The sequence $(Z_n)_{n\in \N}$ is decreasing and converges pointwise to $X$. By Axiom (M), $Z_n\succsim X_m$ for all $n\in \N$ and $m\geq n$. Therefore, $Z_n\succsim \limsup_{m\to \infty}c_{X_m}$ for all $n\in \N$. By Axiom (C), $X\succsim\limsup_{m\to \infty}c_{X_m}$ and $c_X\geq \limsup_{m\to \infty}c_{X_m}$. A similar proof will show that $\liminf_{m\to\infty}c_{X_m}\geq c_X$. Thus, $\lim_{n\to\infty}c_{X_n}=c_X$. 
\end{proof}

\begin{proof}[Proof of Proposition \ref{prop:existance}]
    If $x=y$, this is true; therefore, we can assume that $x>y$. Since $R\in \mathcal{G}$ satisfies $\mathbb{P}(A)\in (0,1)$, by Axiom (SRM), $x\succ xRy\succ y$. Therefore, $x>c_{xRy}>y$. By Axiom (SRM),
    \begin{equation}
    \label{eq:nonEmpty}
        c_{xRx}Rc_{xRy}=xRc_{xRy}\succ xRy\succ c_{xRy}Ry=c_{xRy}Rc_{yRy}.
    \end{equation}
    By Proposition \ref{prop:cert}, the mappings 
    $\phi_1,\phi_2:[y,x]\to[y,x]$ defined by
    $\phi_1(z)=c_{xRz}$ and $\phi_2(z)=c_{zRy}$ are continuous. Furthermore, by Axiom (SRM), $\phi_1$ and $\phi_2$ are strictly increasing. Therefore, by Proposition \ref{prop:cert} and Axiom (SRM), the mapping $\phi:[y,x]\to[y,x]$ defined by 
    $\phi(z)=c_{\phi_1(z)R\phi_2(z)}$ is continuous and strictly increasing. By \eqref{eq:nonEmpty}, $\phi(x)>c_{xRy}>\phi(y)$. Therefore, by the intermediate value theorem, there exists $z'\in [x,y]$ such that $\phi(z')=c_{xRy}$ or  $c_{xRz'}Rc_{z'Ry}\simeq xRy$. The uniqueness of $z'$ follows from the fact that $\phi$ is strictly increasing.
\end{proof}

\begin{proposition}
    \label{prop:meas}
    Let $\succsim$ satisfy Axioms (RC), (M), (SRM), (C). Then for all $X,Y\in \mathcal{X}$, $$(1/2)X\oplus(1/2)Y\in \mathcal{X}.$$
\end{proposition}

\begin{proof}Since there exists $M>0$ such that $-M\leq X,Y\leq M$, we have that $$-M\leq (1/2)X(\omega)\oplus (1/2)Y(\omega)\leq M,~~\omega\in \Omega.$$ Therefore, $(1/2)X\oplus (1/2)Y$ is bounded.

If the binary operation $\oplus$ is measurable, then the $(1/2)X\oplus(1/2)Y$ is measurable as it is then the composition of measurable mappings. To show measurability, we will show that the binary operation $\oplus$ is continuous. Let  $((x_n,y_n))_{n\in \N}\subseteq \mathbb{R}^2$ be a convergent sequence converging to $(x,y)\in \mathbb{R}^2$. We claim that there exists a subsequence $((x_{n_k},y_{n_k}))_{k\in \N}$ such that $(x_{n_k}\oplus y_{n_k})_{k\in \N}$ converges to $x\oplus y$.

If $x=y$, we claim that we can take the original sequence $((x_n,y_n))_{n\in \N}$. Let $\epsilon>0$, there exists $N\in \N$ such that for all $n\geq N$, $x-\epsilon<x_n<x+\epsilon $ and $x-\epsilon<y_n<x+\epsilon $. 
Therefore, $x-\epsilon<x_n\oplus y_n<x+\epsilon$ for all $n\in \N$, and $\lim_{n\to\infty}(x_n\oplus y_n)= x=x\oplus y$.

If $x\neq y$, without loss of generality, we may assume $x>y$ and $x_n>y_n$ for all $n\in \N$. Since the sequence $((x_n,y_n))_{n\in \N}$ is convergent, it is bounded. Therefore, the sequence $(x_n\oplus y_n)_{n\in \N}$ is bounded. By the Bolzano–Weierstrass theorem, there exists a subsequence $(x_{n_k}\oplus y_{n_k})_{k\in \N}$ that converges to some $z\in \mathbb{R}$. Fix $R\in \mathcal{G}_{\mathrm{cf}}$, to ease notation, for $k\in \N$, let $a_k,b_k,d_k\in \mathbb{R}$ satisfy $$a_k\simeq x_{n_k}Ry_{n_k},~~ b_k\simeq x_{n_k}Rz_{n_k},~~ \text{and} ~~d_k\simeq z_{n_k}Ry_{n_k}, $$
where $z_{n_k}=x_{n_k}\oplus y_{n_k}$.
Furthermore, let $b,d\in \mathbb{R}$ satisfy $b\simeq xRz$ and $d\simeq zRy$.
By the definition of $\oplus$, we know that $a_k\simeq b_kRd_k$ for all $k\in \N$. Since $(x_{n_k}Ry_{n_k})_{k\in \N}$, $(x_{n_k}Rz_{n_k})_{k\in \N}$, and $(z_{n_k}Ry_{n_k})_{k\in \N}$ are bounded, by  Proposition \ref{prop:cert}, 
$\lim_{k\to\infty}a_k=c_{xRy}$, $\lim_{k\to\infty}b_k=b$, and $\lim_{k\to\infty}d_k=d$. Furthermore, as $(b_kRd_k)_{k\in \N}$ is bounded, by Proposition \ref{prop:cert}, $\lim_{k\to\infty}c_{b_kRd_k}=c_{bRd}$. Therefore,
$$xRy\simeq c_{xRy}=\lim_{k\to\infty}a_k=\lim_{k\to\infty}c_{b_kRd_k}=c_{bRd}\simeq bRd= c_{xRz}Rc_{zRy}.$$ Thus, $\lim_{k\to\infty}(x_{n_k}\oplus y_{n_k})=x\oplus y$.

Assume for the sake of contradiction that $\oplus$ was not continuous. Then there would exist a convergent sequence $((x_n,y_n))_{n\in \N}\subseteq \mathbb{R}^2$ converging to $(x,y)\in \mathbb{R}^2$ and $\epsilon>0$, such that for all $n\in \N$, $|(x_n\oplus y_n)-(x\oplus y)|\geq \epsilon$. By moving to a subsequence as described above, we get a contradiction. Therefore $\oplus$ is continuous.  
\end{proof}

In the following proof, $\mathcal{X}_0$ denotes the collection of simple acts.

\begin{proof}[Proof of Theorem \ref{theo:good}]
    Assume that $\succsim$ satisfies Axioms (RC), (M), (SRM), (C), (RS), and (SCI). Consider Axioms (A1)-(A6) and (A9) in \cite{GMMS01A}. It is clear that $\succsim$ satisfies Axioms (A1) and (A2). Fix $R\in \mathcal{G}_{\mathrm{cf}}$. By Axiom (SRM), for all $x,y\in \mathbb{R}$ satisfying $x>y$, it holds that $x\succ xRy\succ y$. Therefore, $R$ is essential as defined in \cite{GMMS01A}, and $\succsim$ satisfies Axiom (A3). Furthermore, it is clear that Axiom (SRM) also implies that $\succsim$ satisfies Axiom (A4). Let $(a_\upsilon)_{\upsilon\in \Upsilon}\subseteq \mathbb{R}$ and $(b_\upsilon)_{\upsilon\in \Upsilon}\subseteq \mathbb{R}$ be nets converging to $a,b\in \mathbb{R}$ respectively. Assume that $X\in \mathcal{X}$ satisfies 
   $a_\upsilon Rb_\upsilon\succsim X$ for all $\upsilon\in \Upsilon$. It is straightforward to show that existence of a mapping $\phi:\mathbb{N}\to\Upsilon$ such that $\phi(n)\leq \phi(n+1)$ for all $n\in \N$, $\lim_{n\to\infty}a_{\phi(n)}=a$, and $\lim_{n\to\infty}b_{\phi(n)}=b.$
   Therefore, as the sequence $\left(a_{\phi(n)}Rb_{\phi(n)}\right)_{n\in \N}$ is bounded, by Axiom (C), it holds that $aRb\succsim X$. A similar argument will shows that $X\succsim a_\upsilon Rb_\upsilon$ for all $\upsilon\in \Upsilon$ implies $X\succsim aRb$. Therefore $\succsim$ satisfies Axiom (A5). Let $x,y,z,z'\in \mathbb{R}$ satisfy $x\geq y$ and $z,z'\in [y,x]$. Then $\succsim$ satisfies Axiom (A6) if, for all $A,B\in \{R,R^c,\varnothing, \Omega\},$ it holds that 
   \begin{equation}
        \label{eq:check}
       c_{xAz}Bc_{z'Ay}\simeq c_{xBz'}Ac_{zBy}.
   \end{equation} It is straightforward to show that \eqref{eq:check} holds true if $A\in \{\varnothing, \Omega\}$ or $B\in \{\varnothing, \Omega\}$. If $A=B=R$ or $A=B=R^c$, then \eqref{eq:check} is the same statement as Axiom (RS). By Axiom (RC), we have $c_{xRz}Rc_{z'Ry}\simeq c_{xRz}R^cc_{z'Ry}$, $c_{xRz'}=c_{xR^cz'}$, and $c_{zRy}=c_{zR^cy}$. Therefore, by Axiom (RS),
   $$c_{xRz}R^cc_{z'Ry}\simeq c_{xRz}Rc_{z'Ry}\simeq c_{xRz'}Rc_{zRy} = c_{xR^cz'}R c_{zR^cy}.$$ Thus, for $A=R$ and $B=R^c$, \eqref{eq:check} holds. The case of $A=R^c$ and $B=R$ is identical. Therefore, $\succsim$ satisfies Axiom (A6). Finally, the subjective mixture operation for acts $\oplus$, as defined in our paper, coincides with the definition from \cite{GMMS01A}. Therefore, Axiom (SCI) is equivalent to Axiom (A9).

   Since $\succsim$ satisfies Axioms (A1)-(A6) and (A9), by \citet[Proposition 9]{GMMS01A}, there exists a continuous vNM utility $u$ and a capacity $\tilde{\nu}$ such that
   $$X\succsim Y\iff \int_\Omega u(X)\d\tilde{\nu}\geq \int_\Omega u(Y)\d\tilde{\nu},~~X,Y\in \mathcal{X}_0.$$
   By Axiom (SRM), the function $u$ must be strictly increasing. Therefore, we may assume without loss of generality that $u(0)=0$ and $u(1)=1$. Fix the distortion function $\varphi$ such that $\varphi(u(x))=u(\varphi(x))=x$ for all $x\in [0,1]$. It is clear that $\varphi$ is continuous, strictly increasing, and that $c_A=\varphi(\tilde{\nu}(A))$ for all $A\in \mathcal{F}$, where $c_A=c_{\id_A}$. Let $(A_n)_{n\in \mathbb{N}}\subseteq\mathcal{F}$ be an increasing sequence and define $A=\bigcup_{n\in \N}A_n$. By Proposition \ref{prop:cert},
    $$\varphi(\tilde{\nu}(A))= c_{A}=\lim_{n\to\infty}c_{A_n}=\lim_{n\to\infty}\varphi(\tilde{\nu}(A_n))=\varphi\left(\lim_{n\to\infty}\tilde{\nu}(A_n)\right).$$
    Therefore, $\tilde{\nu}$ is upward continuous. A similar proof will show that $\tilde{\nu}$ is downward continuous and, therefore, $\tilde{\nu}$ is continuous.

   Let $R_1,R_2\in \mathcal{G}$ satisfy $\mathbb{P}(R_1)=\mathbb{P}(R_2)$. By Axiom (RC), $R_1\simeq R_2$ and
   $\tilde{\nu}(R_1)=\tilde{\nu}(R_2)$. Therefore, there exists a distortion function $g$ such that for all $R\in\mathcal{G}$, $\tilde{\nu}(R)=g(\mathbb{P}(R)).$ Since $\tilde{\nu}$ is continuous, $g$ is continuous. Furthermore, by Axiom (SRM), $g$ is strictly increasing. Denote by $h$ the continuous and strictly increasing distortion function satisfying $g(h(x))=h(g(x))=x$ for all $x\in [0,1]$. Define the capacity $\nu=h\circ \tilde{\nu}$. Clearly $\nu$ is a continuous risk-conforming capacity. Therefore \eqref{eq:rep} holds for $\mathcal{X}_0$.

   Finally, we need to show that \eqref{eq:rep} holds for $\mathcal{X}$. Given $X,Y\in \mathcal{X}$ satisfying $X\succsim Y$, fix the bounded sequences $(X_n)_{n\in \mathbb{N}}\subseteq \mathcal{X}_0$ and $(Y_n)_{n\in \mathbb{N}}\subseteq \mathcal{X}_0$ satisfying $X_n\geq X$ for all $n\in \N$, $Y\geq Y_n$ for all $n\in \N$, $X_n\to X$ pointwise, and $Y_n\to Y$ pointwise. Since $X$ and $Y$ are bounded, the sequences $(X_n)_{n\in \mathbb{N}}$ and $(Y_n)_{n\in \mathbb{N}}$ exist. By Axiom (M), it holds that for all $n\in \mathbb{N}$, $X_n\succsim Y_n$. Therefore, for all $n\in \mathbb{N}$,
   $$\int_{\Omega} u(X_n)\d(g\circ \nu)\geq \int_{\Omega} u(Y_n)\d(g\circ \nu).$$
   Because $u(X_n)\to u(X)$ pointwise, $u(Y_n)\to u(Y)$ pointwise, $(u(X_n))_{n\in \N}$ is bounded, $(u(Y_n))_{n\in \N}$ is bounded, and $g\circ\nu$ is continuous, by using \citet[Theorem 11.9]{WK10A},  we have 
   $$\int_{\Omega}u(X)\d(g\circ \nu)=\lim_{n\to\infty}\int_{\Omega} u(X_n)\d(g\circ \nu)\geq \lim_{n\to\infty}\int_{\Omega} u(Y_n)\d(g\circ \nu)=\int_{\Omega}u(Y)\d(g\circ\nu).$$ Given $X,Y\in \mathcal{X}$ satisfying 
   $$\int_{\Omega}u(X)\d(g\circ \nu)\geq \int_{\Omega}u(Y)\d(g\circ\nu),$$ fix the same bounded sequences $(X_n)_{n\in \mathbb{N}}\subseteq \mathcal{X}_0$ and $(Y_n)_{n\in \mathbb{N}}\subseteq \mathcal{X}_0$ as before. By the monotonicity of the Choquet integral, we have that for all $m,n\in\N$,
   $$\int_{\Omega} u(X_m)\d(g\circ \nu)\geq \int_{\Omega} u(Y_n)\d(g\circ \nu).$$
   Therefore, for all $m,n\in\N$, $X_m\succsim Y_n$. Since  $(Y_n)_{n\in \N}$ is bounded and $Y_n\to Y$ pointwise, by Axiom (C), $X_m\succsim Y$ for all $m\in \N$. Since  $(X_m)_{m\in \N}$ is bounded and $X_m\to X$ pointwise, by Axiom (C), $X\succsim Y$. Thus, \eqref{eq:rep} holds.

   Conversely, assume that $\succsim$ is a CRDU preference relation. The fact that $\succsim$ satisfies Axioms (RC), (M), (SRM), and (C) follows from Proposition \ref{prop:consis}.
   
   Let $R\in \mathcal{G}_{\mathrm{cf}}$ and $x,y,z,z'\in \mathbb{R}$ be chosen such that $x\geq y$ and $z,z'\in [y,x]$. We have 
   $$u(c_{xRz})=u(x)g(1/2)+u(z)(1-g(1/2))\geq u(z')g(1/2)+u(y)(1-g(1/2))=u(c_{z'Ay}).$$
   Therefore, $c_{xAz}\geq c_{z'Ay}$. Using a similar argument, one can show that $c_{xRz'}\geq c_{zRy}$. Thus, 
   \begin{align*}
       & \int_{\Omega}u(c_{xRz}Rc_{z'Ry})\d(g\circ\nu)\\&=u(c_{xRz})g(1/2)+u(c_{z'Ry})(1-g(1/2))\\&=\left[u(x)g(1/2)+u(z)(1-g(1/2))\right]g(1/2) +\left[u(z')g(1/2)+u(y)(1-g(1/2))\right](1-g(1/2))\\&=\left[u(x)g(1/2)+u(z')(1-g(1/2))\right]g(1/2) +\left[u(z)g(1/2)+u(y)(1-g(1/2))\right](1-g(1/2))\\&=u(c_{xRz'})g(1/2)+u(c_{zRy})(1-g(1/2))=\int_{\Omega}u(c_{xRz'}Rc_{zRy})\d(g\circ\nu),
   \end{align*}
   that is, 
   $c_{xRz}Rc_{z'Ry}\simeq c_{xRz'}Rc_{zRy}$, and hence $\succsim$ satisfies Axiom (RS).

   Let $R\in \mathcal{G}_{\mathrm{cf}}$, $x,y\in \mathbb{R}$ with $x\geq y$, and $z=x\oplus y$. It holds that
   \begin{align*}
    &\left[u(x)g(1/2)+u(z)(1-g(1/2))\right]g(1/2)+\left[u(z)g(1/2)+u(y)(1-g(1/2))\right](1-g(1/2))\\&=u(c_{xRz})g(1/2)+u(c_{zRy})(1-g(1/2))\\&=u(c_{c_{xRz}Rc_{zRy}})\\&=u(c_{xRy})= u(x)g(1/2)+u(y)(1-g(1/2)).  \end{align*} 
    Therefore,
    \begin{align*}
    u(x)g(1/2)(1-g(1/2))+u(y)g(1/2)(1-g(1/2)) =2u(z)g(1/2)(1-g(1/2)).
   \end{align*}
   As $1>g(1/2)>0$, it holds that $u(z)=(1/2)u(x)+(1/2)u(y).$ Therefore, $u((1/2)X\oplus_R(1/2)Y)=(1/2)u(X)+(1/2)u(Y)$ for all $X,Y\in \mathcal{X}$. Using the comonotonic additivity of the Choquet integral, it is clear that $\succsim$ satisfies Axiom (SCI). 
\end{proof}

\section{Proofs accompanying Section \ref{sec:prop}}
\label{App:prop}

\begin{proof}[Proof of Theorem \ref{prop:comAm}]
    Assume that $\succsim_2$ is more ambiguity averse than $\succsim_1$. Let $A\in \mathcal{F}$ and $R_1\in \mathcal{G}$ such that $A\simeq_1 R_1$. By the definition of comparative ambiguity attitudes, $R_1\succsim_2 A$. Find $R_2\in \mathcal{G}$ such that $A\simeq_2 R_2$. As $R_1 \succsim_2 R_2$, it holds that $\mathbb{P}(R_1)\geq \mathbb{P}(R_2)$ and $\nu_1(A)\geq \nu_2(A).$ To show the converse, let $R\in \mathcal{G}$ and $A\in \mathcal{F}$ satisfy $R\succsim_1 A$. Furthermore, let $R_1,R_2\in \mathcal{G}$ satisfy $A\simeq_1 R_1$ and $A\simeq_2 R_2$. Since $\nu_1(A)\geq \nu_2(A)$ and $R\succsim_1 A$, $\mathbb{P}(R)\geq\mathbb{P}(R_1)\geq \mathbb{P}(R_2)$. Therefore, $R\succsim_2 R_2\simeq_2 A$.
\end{proof}

\begin{proof}[Proof of Proposition \ref{prop:attitudes}]
    Assume that the $\succsim_1$-utility is equal to the $\succsim_2$-utility, the $\succsim_1$-distortion is equal to the $\succsim_2$-distortion, and $\succsim_2$ is more ambiguity averse than $\succsim_1$. Denote the shared utility by $u$ and the shared distortion by $g$. Denote the $\succsim_1$-matching probability by $\nu_1$ and the $\succsim_2$-matching probability by $\nu_2$. By Theorem \ref{prop:comAm}, for all $A\in \mathcal{F}$, $g(\nu_1(A))\geq g(\nu_2(A))$. Let $X\in \mathcal{X}(\mathcal{G})$ and $Y\in \mathcal{X}$ satisfy $X\succsim_1Y~(X\succ_1Y)$. We have 
    $$\int_\Omega u(X)\d(g\circ\nu_2)=\int_\Omega u(X)\d(g\circ\mathbb{P})=\int_\Omega u(X)\d(g\circ\nu_1)\overset{(>)}{\geq} \int_{\Omega}Y\d(g\circ\nu_1)\geq \int_{\Omega}Y\d(g\circ\nu_2).$$
    Therefore, $X\succsim_2 Y~(X\succ_2Y)$.

    For the converse, assume that \eqref{eq:ep} holds true. Since \eqref{eq:ep} is stronger than the definition of comparative ambiguity aversion, $\succsim_2$ is more ambiguity averse than $\succsim_1$. Let $X,Y\in \mathcal{X}(\mathcal{G})$. If $X\succsim_1 Y$, then $X\succsim_2 Y$ by \eqref{eq:ep}.
    If $X\succsim_2 Y$, then $X\succsim_1 Y$ or else there would be a contradiction with \eqref{eq:ep}. Thus,
    $$X\succsim_1 Y\iff X\succsim_2 Y,~~X,Y\in \mathcal{X}(\mathcal{G}).$$
    Denote the $\succsim_1$-utility by $u_1$, the $\succsim_1$-distortion by $g_1$, the $\succsim_2$-utility by $u_2$, and the $\succsim_2$-distortion by $g_2$. Therefore, the functions
    $$I_1(X)=\int_\Omega u_1(X)\d(g_1\circ\mathbb{P})\text{ and }I_2(X)=\int_\Omega u_2(X)\d(g_2\circ\mathbb{P}),~~X\in \mathcal{X}(\mathcal{G}),$$
    both represent the preference relation on $\mathcal{X}(\mathcal{G})$. By \citet[Theorem 11]{GM01A}, as $\succsim$ is a biseperable preference relation, 
$u_1=au_2+b$ for some $a>0$ and $b\in \mathbb{R}$, and $g_1\circ\mathbb{P}=g_2\circ\mathbb{P}.$ Clearly, this means that $g_1=g_2$. As $u_1(0)=u_2(0)=0$ and $u_1(1)=u_2(1)=1$, we have that $u_1=u_2$.
\end{proof}

The next lemma is a variant of a well-known result; see Lemma 6 of \cite{AW20A}. Although their result assumes some additional properties, Property (AN) is sufficient for the proof of their part (ii) to go through, which is what the following lemma states. 
\begin{lemma} \label{lem:need}
    Let $\succsim$ satisfy Property (AN). Then for all $A_1,A_2,B_1,B_2\in \mathcal{F}$ satisfying $A_1\cap A_2=\varnothing$, $B_1\cap B_2=\varnothing$, and $A_1\succsim B_1$, it holds that $$A_2\succsim B_2~(\mbox{resp.~} A_2\succ B_2)\implies A_1\cup A_2\succsim B_1\cup B_2~(\mbox{resp.~} A_1\cup A_2\succ B_1\cup B_2).$$ 
\end{lemma} 

\begin{proof}[Proof of Theorem \ref{th:AN}]
    Let $\nu$ denote the $\succsim$-matching probability, whose existence is given by Proposition \ref{proposition:match}. Assume that $\succsim$ satisfies Property (AN). Let $A,B\in \mathcal{F}$ be disjoint. For the sake of contradiction, assume that  $\nu(A)+\nu(B)> 1$. Let $R_A\in \mathcal{G}$ satisfy $A\sim R_A$. Therefore, $\nu(A)=\mathbb{P}(R_A)$. Since $\nu(B)>1-\nu(A)=\mathbb{P}(R^c_A)$, we have, by Proposition \ref{proposition:match}, $B\succ R^c_A$. Therefore, by Lemma \ref{lem:need} and Axiom (MO), 
    $\Omega \succsim A\cup B\succ R_A\cup R^c_A=\Omega,$ which is a contradiction.
    
    Since $\nu(A)+\nu(B)\leq 1$, we can find disjoint $R_A,R_B\in \mathcal{G}$ such that $A\sim R_A$ and $B\sim R_B$.
    Therefore, by Lemma \ref{lem:need}, we have $A\cup B\sim R_A\cup R_B$. Thus,
    $$\nu(A\cup B)=\mathbb{P}(R_A\cup R_B)=\mathbb{P}(R_A)+\mathbb{P}(R_B)=\nu(A)+\nu(B).$$
    Since $\nu$ is continuous (Proposition \ref{proposition:match}), we have that $\nu\in \mathcal{M}_1$.

    Conversely, assume that $\nu\in \mathcal{M}_1$. Let $A,B,C\in \mathcal{F}$ satisfy $(A\cup B)\cap C=\varnothing.$ By Proposition \ref{proposition:match},
    $$A\succsim B\iff \nu(A)\geq \nu(B)\iff \nu(A\cup C)\geq \nu(B\cup C)\iff A\cup C\succsim B\cup C.$$
    Thus, $\succsim$ satisfies (AN).
\end{proof}

In the following proof, $\mathcal{X}_0(\mathcal{G}) = \mathcal{X}_0 \cap \mathcal{X}(\mathcal{G})$, where $\mathcal{X}_0$ denotes the collection of simple acts. 

\begin{proof}[Proof of Proposition \ref{prop:FSD}]
    Let $\nu$ denote the $\succsim$-matching probability. Assume that Property (AN) holds. Therefore, by Theorem \ref{th:AN}, $\nu\in \mathcal{M}_1$. Let $X,Y\in \mathcal{X}$ satisfy $X\geq_{\nu}^{\mathrm{sd}}Y$. We have that $$g(\nu(u(X)>x))\geq g(\nu(u(Y)>x)),~~x\in \mathbb{R},$$ where $u$ is the $\succsim$-utility and $g$ is the $\succsim$-distortion. Therefore, $\int_{\Omega}u(X)\d(g\circ \nu)\geq \int_{\Omega}u(Y)\d(g\circ \nu)$ and $X\succsim Y$. Thus $\succsim$ satisfies \eqref{eq:FSD} with $\mu=\nu$.

    Assume that \eqref{eq:FSD} holds with $\mu\in \mathcal{M}_1$. We claim that 
    $$\mathbb{P}(A)=0\iff \mu(A)=0,~~A\in \mathcal{G}.$$
    Define $\epsilon_{\mu}=\inf\left\{\mu(B):B\in \mathcal{G}\text{ and }B\succ \varnothing\right\}.$ If $B\in \mathcal{G}$ satisfies $\mu(B)>\epsilon_{\mu}$, then there exists $B'\in \mathcal{G}$ such that $\mu(B)>\mu(B')$ and $B\succ\varnothing$. As $\succsim$ satisfies Property \eqref{eq:FSD}, this implies that $B\succsim B'\succ\varnothing$. Furthermore, if $B\in \mathcal{G}$ satisfies $B\succ\varnothing$, then $\mu(B)\geq \epsilon_{\mu}$. For the sake of contradiction, assume that $\epsilon_{\mu}>0$ and fix $n\in\N$ satisfying $n>1/\epsilon_{\mu}$. Since $(\Omega,\mathcal{G},\mathbb{P})$ is atomless, we can find a partition of $\Omega$, $\{A_1,\dots,A_n\}\subseteq \mathcal{G}$, such that $\mathbb{P}(A_k)=1/n$ for each $k\in \{1,\dots,n\}.$ By Axiom (SRM), $A_k\succ\varnothing$ for all $k\in \{1,\dots,n\}$.
    Therefore $\mu(A_k)\geq \epsilon_{\mu}$ for all $k\in \{1,\dots,n\}$. Thus 
    $$1=\mu(\Omega)=\sum_{k=1}^n \mu(A_k)\geq n\epsilon_{\mu}>1,$$
    a contradiction. Therefore, $\epsilon_{\mu}=0$.
    For all $A\in \mathcal{G}$, by Axiom (SRM), we have
    $$\mathbb{P}(A)=0\iff A\simeq \varnothing\iff \mu(A)\leq \epsilon_{\mu}\iff \mu(A)=0.$$
    
    Since $(\Omega,\mathcal{G},\mathbb{P})$ is atomless, $(\Omega,\mathcal{G},\mu|_{\mathcal{G}})$ is atomless. Define the function $\psi:\mathcal{X}_0(\mathcal{G})\to \mathbb{R}$ by
    $$\psi(X)=\int_{\Omega} u(X)d(g\circ \mathbb{P}),$$
    where $u$ is the $\succsim$-utility and $g$ is the $\succsim$-distortion. For the sake of contradiction, assume that $\mathbb{P}\neq \mu|_{\mathcal{G}}$. Then, by \cite[Theorem 6]{L24A}, $\psi$ must satisfy
    $$\psi(X)=T\left(\max\{X(\omega):\omega\in \Omega\},\min\{X(\omega):\omega\in \Omega\}\right),$$
    where $T:\mathbb{R}^2\to\mathbb{R}$. Let $\{R_1,R_2,R_3\}\subseteq \mathcal{G}$ be a partition of $\Omega$ satisfying $\mathbb{P}(R_1)=\mathbb{P}(R_2)=\mathbb{P}(R_3)=1/3$ and $x_1,x_2,x_3,x_4\in \mathbb{R}$ satisfy $x_1>x_2>x_3>x_4$. Therefore, we have
    \begin{align*}
        T(x_1,x_4)=\psi(x_1R_1x_2R_2x_4)&=\int_{\Omega}u(x_1R_1x_2R_2x_4)\d(g\circ \mathbb{P})\\&>\int_{\Omega}u(x_1R_1x_3R_2x_4)\d(g\circ \mathbb{P})=\psi(x_1R_1x_3R_2x_4)=T(x_1,x_4),
    \end{align*}
    which is a contradiction. Thus, $\mathbb{P}= \mu|_{\mathcal{G}}$. Let $A\in \mathcal{F}$ and $R_A\in \mathcal{G}$ such that $\mu(A)=\mathbb{P}(R_A)$. Since $\succsim$ satisfies \eqref{eq:FSD}, we have $A\simeq R_A$. Therefore $\nu(A)=\mathbb{P}(R_A)=\mu(A)$. Thus $\nu=\mu$ and $\nu\in \mathcal{M}_1$. By Theorem \ref{th:AN}, $\succsim$ satisfies Property (AN).
\end{proof}

\begin{proof}[Proof of Proposition \ref{prop:UA}]
    Assume that $\succsim$ satisfies Property (DS). Let $X,Y\in \mathcal{X}$ satisfy $X\succsim Y$. We claim that there exists $c_0\geq 0$ such that $X-c_0\simeq Y.$ Define $\phi:[0,\infty)\to \mathbb{R}$ by 
    $$\phi(c)=\int_{\Omega}u(X-c)\d(g\circ\nu).$$
     By \citet[Theorem 11.9]{WK10A}, the function $\phi$ is continuous. Define the constant $c_1=\|X\|+\|Y\|.$ It holds that
    $X-c_1\leq -\|Y\|\leq Y.$ By Axiom (M), $Y\succsim X-c_1$ and $$\phi(0)=\int_{\Omega}u(X)\d(g\circ\nu)=\int_{\Omega}u(Y)\d(g\circ\nu)\geq \int_{\Omega}u(X-c_1)\d(g\circ\nu)=\phi(c_1).$$
    Therefore, by the intermediate value theorem, there exists $c_0\geq 0$ such that $X-c_0\simeq Y$.
    By Property (DS) and Axiom (M), for all $\lambda\in [0,1]$,
    $$\lambda X+(1-\lambda)Y\succsim \lambda (X-c_0)+(1-\lambda)Y\succsim Y.$$
    Therefore, $\succsim$ is convex as defined in \citet[Definition 1]{WY19A}. By \citet[Corollary 7]{WY19A}, as $u$ is strictly increasing, $u$ is concave and $g\circ \nu$ is supermodular. The proof of the converse is straightforward.
\end{proof}

\begin{proof}[Proof of Theorem \ref{prop:main}]
    This is a consequence of Proposition \ref{prop:latt} in Appendix \ref{ap:mod}, Proposition \ref{prop:UA}, and the remark regarding the work of \cite{SZ08A}. 
\end{proof}

In the following proof, given $A\in \mathcal{F}$, we use $\mathbb{Q}(A|\mathcal{H})$ to denote the conditional probability of $A$ given the sub-$\sigma$-algebra $\mathcal{H}$. 

\begin{proof}[Proof of the claim in Example \ref{ex:counter}]
    Define the capacity 
    $$\tilde{\nu}(A)=\int_{\Omega} g(\mathbb{Q}(A|\mathcal{H}))\d\mathbb{Q},~~A\in \mathcal{F}.$$
    An application of both the monotone convergence theorem and the conditional monotone convergence theorem will verify that $\tilde{\nu}$ is continuous. As $\mathcal{G}$ and $\mathcal{H}$ are independent, it is straightforward to show that for all $A\in \mathcal{G}$, $\tilde{\nu}(A)=g(\mathbb{Q}(A))=g(\mathbb{P}(A))$. Also, for all $A\in \mathcal{H}$, $\tilde{\nu}(A)=\mathbb{Q}(A)$. Since $g$ is convex (as $g$ is the inverse of $h$), $g(x)\leq x$ for all $x\in [0,1]$. Therefore, $\tilde{\nu}(A)\leq \mathbb{Q}(A)$ for all $A\in \mathcal{F}$. Finally, by Jensen's inequality, it holds that 
    $$\tilde{\nu}(A)=\int_{\Omega} g(\mathbb{Q}(A|\mathcal{H}))\d\mathbb{Q}\geq g\left(\int_{\Omega}\mathbb{Q}(A|\mathcal{H})\d\mathbb{Q}\right)=g(\mathbb{Q}(A)),~~A\in \mathcal{F}.$$
    Given $A,B\in \mathcal{F}$, as $\id_{A\cup B}+\id_{A\cap B}=\id_A+\id_B$,
    $\mathbb{Q}(A\cup B|\mathcal{H})+\mathbb{Q}(A\cap B|\mathcal{H})=_{\mathbb{Q}}^{\mathrm{as}}\mathbb{Q}(A|\mathcal{H})+\mathbb{Q}(B|\mathcal{H}).$ Using a similar argument as \citet[Proposition 4.75]{FS16A}, it holds that
    \begin{equation}
        \label{eq:sup1}
        g(\mathbb{Q}(A\cup B|\mathcal{H}))+g(\mathbb{Q}(A\cap B|\mathcal{H}))\geq_{\mathbb{Q}}^{\mathrm{as}} g(\mathbb{Q}(A|\mathcal{H}))+g(\mathbb{Q}(B|\mathcal{H})).
    \end{equation}
    Applying the Choquet integral with respect to $\mathbb{Q}$ to \eqref{eq:sup1} shows that $\tilde{\nu}$ is supermodular. 
    
    Define the capacity $\nu$ by $\nu(A)=h(\tilde{\nu}(A))$ for all $A\in \mathcal{F}.$ It is clear that $\nu$ is risk conforming, $\nu(A)=h(\mathbb{Q}(A))$ for all $A\in \mathcal{H}$, $h(\mathbb{Q}(A))\geq \nu(A)\geq \mathbb{Q}(A)$ for all $A\in \mathcal{F}$. Since $g\circ\nu=\tilde{\nu}$, $g\circ\nu$ is supermodular. As $h$ is continuous, $\nu$ is continuous.
\end{proof}

\begin{proof}[Proof of Theorem \ref{prop:maxmin}]
    Assume that $\nu$ is supermodular. Let $X\in \mathcal{X}$, by \citet[Theorem 4.7]{MM04A}, as $\{u(X)>x\}_{x\in \mathbb{R}}$ is a chain, there exists $\tilde\mu\in \mathfrak{C}(\nu)$ such that 
    $\tilde\mu(u(X)>x)=\nu(u(X)>x)\leq\mu(u(X)>x)$
    for all $x\in \mathbb{R}$ and $\mu\in \mathfrak{C}(\nu)$. Therefore, as $g$ is increasing, $$\int_{\Omega}u(X)\d(g\circ\tilde\mu)=\int_{\Omega}u(X)\d(g\circ\nu)\leq \int_{\Omega}u(X)\d(g\circ\mu),~~\mu\in \mathfrak{C}(\nu).$$
    Thus, $\min_{\mu\in \mathfrak{C}(\nu)}\int_{\Omega}u(X)\d(g\circ\mu)=\int_{\Omega}u(X)\d(g\circ\tilde\mu)=\int_{\Omega}u(X)\d(g\circ\nu)$. Therefore, for all $X,Y\in \mathcal{X}$,
    $$X\succsim Y\iff \int_{\Omega}u(X)\d(g\circ\nu) \geq \int_{\Omega}u(Y)\d(g\circ\nu)  \iff \min_{\mu\in \mathfrak{C}(\nu)}\int_{\Omega}u(X)\d(g\circ\mu)\geq \min_{\mu\in \mathfrak{C}(\nu)}\int_{\Omega}u(Y)\d(g\circ\mu).$$
    
    Conversely, assume that \eqref{eq:MM} holds. Let $X\in \mathcal{X}$, we know that $u(c_X)=\int_{\Omega}u(X)\d(g\circ \nu)$. Therefore, by \eqref{eq:MM}, $$\int_{\Omega}u(X)\d(g\circ\nu)=\min_{\mu\in \mathfrak{C}(\nu)}\int_{\Omega}u(X)\d(g\circ\mu).$$ Let $A,B\in \mathcal{F}$ satisfy $A\subseteq B$ and $c\in \mathbb{R}$ satisfy $u(c)=2$, by the above result, there exists $\tilde{\mu}\in \mathfrak{C}(\nu)$ such that
    \begin{align*}
        g(\nu(A))+g(\nu(B))&=\int_{\Omega}u(\id_A)\d(g\circ \nu)+\int_{\Omega}u(\id_B)\d(g\circ \nu)\\&=\int_{\Omega}u(\id_A)+u(\id_B)\d(g\circ \nu)\\&=\int_{\Omega}u(c\id_A+\id_{B\setminus A})\d(g\circ \nu)\\&=\int_{\Omega}u(c\id_A+\id_{B\setminus A})\d(g\circ \tilde{\mu})=g(\tilde\mu(A))+g(\tilde\mu(B)).
    \end{align*}
    As $g(\nu(A))\leq g(\tilde\mu(A))$ and $g(\nu(B))\leq g(\tilde\mu(B))$, $g(\nu(A))= g(\tilde\mu(A))$ and $g(\nu(B))= g(\tilde\mu(B))$. Since $g$ is bijective, we have that $\nu(A)=\tilde{\mu}(A)$ and $\nu(B)=\tilde{\mu}(B)$. Therefore, by \citet[Theorem 4.7]{MM04A}, $\nu$ is supermodular.
\end{proof}

\section{Results and proofs accompanying Section \ref{sec:alt}}
\label{app:alt}

\begin{proof}[Proof of Proposition \ref{prop:RAA}]
    Let $\nu$ denote the $\succsim$-matching probability. Assume that $\succsim$ satisfies Property (RAA). Let $A\in \mathcal{F}$ and $R_A\in \mathcal{G}$ such that $A\simeq R_A$. Therefore, $\nu(A)=\mathbb{P}(R_A)$. Furthermore, let $R\in \mathcal{G}$ satisfy $\mathbb{Q}(R)=\mathbb{Q}(A)$. By Property (RAA), $R\succsim A\simeq R_A$. Therefore, by Proposition \ref{proposition:match} and the fact that $\mathbb{Q}$ is risk conforming, $$\mathbb{Q}(A)=\mathbb{Q}(R)=\mathbb{P}(R)\geq \mathbb{P}(R_A)=\nu(A).$$
    Since $A\in \mathcal{F}$ was arbitrary, $\mathbb{Q}\in \mathfrak{C}(\nu).$
    
    Conversely, assume that $\mathbb{Q}\in \mathfrak{C}(\nu)$. Denote by $u$ the $\succsim$-utility and by $g$ the $\succsim$-distortion. Let $X\in \mathcal{X}(\mathcal{G})$ and $Y\in \mathcal{X}$ satisfy $X=_{\mathbb{Q}}^{\mathrm{sd}} Y$. Since $\nu$ and $\mathbb{Q}$ are risk conforming, for all $x\in \mathbb{R}$, we have
    $$g(\nu(u(X)>x))=g(\mathbb{P}(u(X)>x))=g(\mathbb{Q}(u(X)>x))=g(\mathbb{Q}(u(Y)>x))\geq g(\nu(u(Y)>x)).$$
    Therefore, $\int_{\Omega} u(X)\d(g\circ\nu)\geq \int_{\Omega} u(Y)\d(g\circ\nu)$ and $X\succsim Y$. Therefore, $\succsim$ satisfies Property (RAA).
\end{proof}

\begin{proof}[Proof of Proposition \ref{prop:nullity}]
    Let $\nu$ denote the $\succsim$-matching probability. Assume that $\succsim$ satisfies Property (NSC). If $A,B\in \mathcal{F}$ satisfy $\mathbb{Q}(A\triangle B)=0$, then $\id_A =^{\mathrm{as}}_{\mathbb{Q}} \id_B$. By Property (NSC), $A\simeq B$. Therefore, $\nu(A)=\nu(B)$ and $\nu$ is $\mathbb{Q}$-consistent. 
    
    For the converse, assume that $\nu$ is $\mathbb{Q}$-consistent. Denote by $u$ the $\succsim$-utility and by $g$ the $\succsim$-distortion. Fix $X,Y\in \mathcal{X}$ satisfying $X=^{\mathrm{as}}_{\mathbb{Q}}Y$. As $\mathbb{Q}(u(X)\neq u(Y))=0$ and
    $$\{u(X)>x\}\triangle \{u(Y)>x\}\subseteq \{u(X)\neq u(Y)\},$$
    we obtain $\mathbb{Q}(\{u(X)>x\}\triangle \{u(Y)>x\})=0.$ Therefore, since $\nu$ is $\mathbb{Q}$-consistent, we have that $g(\nu(u(X)>x))=g(\nu(u(Y)>x))$ for all $x\in \mathbb{R}$. Thus, 
    $$\int_{\Omega}u(X)\d(g\circ \nu)=\int_{\Omega}u(Y)\d(g\circ \nu),$$
    and we can conclude $X\simeq Y$.
\end{proof}

To ease notation, given a distortion family $(g_X)_{X\in \mathcal{X}}$ and $A\in \mathcal{F}$, we use $g_A$ to denote $g_{\id_A}$.  

\begin{proof}[Proof of Theorem \ref{theo:2}]
    Assume that $\succsim$ is a CRDU preference relation satisfying \eqref{eq:rep} and Property (NSC). Denote by $\nu$ the $\succsim$-matching probability and by $g$ the $\succsim$-distortion. Fix $X\in \mathcal{X}$ and define
    $\tilde{g}_X:\mathcal{D}(X)\to[0,1]$ by
\begin{equation}
\label{eq:def1}\tilde{g}_X(\alpha)=g(\nu(X>x)),~~\text{where}~x\in \mathbb{R}~\text{satisfies}~\alpha=\mathbb{Q}(X>x).
\end{equation}
    To see that $\tilde{g}_X$ is well-defined, let $x,y\in\mathbb{R}$ satisfy $\mathbb{Q}(X>x)=\mathbb{Q}(X>y)$. It is straightforward to verify that
    $\mathbb{Q}(\{X>x\}\triangle\{X>y\})=0$. Therefore, by Proposition \ref{prop:nullity}, we have that $\nu(X>x)=\nu(X>y)$ and $\tilde{g}_X$ is well-defined. Furthermore, if $X\in \mathcal{X}(\mathcal{G})$, as $\nu$ and $\mathbb{Q}$ are risk conforming, it holds that \begin{equation}
        \label{eq:gd}\tilde{g}_X(\mathbb{Q}(X>x))=g(\nu(X>x))=g(\mathbb{P}(X>x))=g(\mathbb{Q}(X>x)), ~~x\in \mathbb{R}.
    \end{equation}
    Thus, if $X\in \mathcal{X}(\mathcal{G})$, $\tilde{g}_X(\alpha)=g(\alpha)$ for all $\alpha\in \mathcal{D}(X)$.
    
    We claim that $\tilde{g}_X$ is increasing. Let $\alpha,\beta\in \mathcal{D}(X)$ with $\alpha< \beta$. By the definition of $\mathcal{D}(X)$, there exists $x,y\in \mathbb{R}$ such that $\alpha=\mathbb{Q}(X>x)$ and $\beta=\mathbb{Q}(X>y)$. As $\alpha<\beta$, it must hold that $y<x$. Therefore, 
    $\tilde{g}_X(\alpha)=g(\nu(X>x))\leq  g(\nu(X>y))=\tilde{g}_X(\beta).$ 
    
    Next, we claim that $\tilde{g}_X$ is continuous. Let $(\alpha_n)_{n\in \N}\subseteq \mathcal{D}(X)$ be a strictly increasing sequence converging to $\alpha\in \mathcal{D}(X)$. There exists a strictly decreasing sequence $(x_n)_{n\in \N}\subseteq\mathbb{R}$ such that $\alpha_n=\mathbb{Q}(X>x_n)$. Define $x=\lim_{n\to\infty}x_n=\inf_{n\in \N}x_n$. It is clear that $\alpha=\mathbb{Q}(X>x)$. We have that 
    $$\lim_{n\to \infty}\tilde{g}_X(\alpha_n)=\lim_{n\to\infty}g(\nu(X>x_n))=g(\nu(X>x))=\tilde{g}_X(\alpha).$$
    Let $(\alpha_n)_{n\in \N}\subseteq \mathcal{D}(X)$ be a strictly decreasing sequence converging to $\alpha\in \mathcal{D}(X)$. There exists a strictly increasing sequence $(x_n)_{n\in \N}\subseteq\mathbb{R}$ such that $\alpha_n=\mathbb{Q}(X>x_n)$. As $\alpha\in \mathcal{D}(X)$, there exists $z\in \mathbb{R}$ such that $\mathbb{Q}(X>z)=\alpha$. Define $x=\lim_{n\to\infty}x_n=\sup_{n\in \N}x_n$, as $z>x_n$ for all $n\in\N$, $z\geq x$. As 
    $$\alpha=\lim_{n\to\infty}\mathbb{Q}(X>x_n)=\mathbb{Q}(X\geq x)\geq \mathbb{Q}(X>x)\geq \mathbb{Q}(X>z)=\alpha,$$
    we have $\mathbb{Q}(X=x)=0$. Therefore, $\mathbb{Q}(\{X>x\}\triangle \{X\geq x\})=0$. By Proposition \ref{prop:nullity}, it holds that $\nu(X>x)=\nu(X\geq x).$ Since $\mathbb{Q}(X>x)=\alpha$, we have
    $$\lim_{n\to \infty}\tilde{g}_X(\alpha_n)=\lim_{n\to\infty}g(\nu(X>x_n))=g(\nu(X\geq x))=g(\nu(X> x))=\tilde{g}_X(\alpha).$$
    Therefore, $\tilde{g}_X$ is continuous.

    If $X=\id_A$ for some $A\in \mathcal{F}$, define $g_A:[0,1]\to [0,1]$ by $g_A(\alpha)= \tilde{g}_A(\mathbb{Q}(A))$ for $\alpha\in (0,1)$, $g_A(0)=0$, and $g_A(1)=1$. It is clear that $g_A$ is a distortion function. Furthermore, as $\mathcal{D}(\id_A)=\{0,\mathbb{Q}(A),1\}$, $g_A(\alpha)=\tilde{g}_A(\alpha)$ for all $\alpha\in \mathcal{D}(\id_A)$. If $X\neq \id_A$ for some $A\in \mathcal{F}$, define $g_X:[0,1]\to[0,1]$ by $$g_X(\alpha)=\sup\{\tilde{g}_X(\beta):\beta\in \mathcal{D}(X)\text{ and }\beta\leq \alpha\}.$$ It is straightforward to verify that $g_X$ is a distortion function. Since $\tilde{g}_X$ is increasing, $g_X(\alpha)=\tilde{g}_X(\alpha)$ for all $\alpha\in \mathcal{D}(X)$.

    Since for all $X\in \mathcal{X}$, $g_X(\alpha)=\tilde{g}_X(\alpha)$ for all $\alpha\in \mathcal{D}(X)$, $(g_X)_{X\in \mathcal{X}}$ satisfies (b). For $A,B\in \mathcal{F}$ satisfying $A\subseteq B$ and $\alpha \in (0,1)$, we have
    $$g_A(\alpha)=\tilde{g}_A(\mathbb{Q}(A))=g(\nu(A))\leq g(\nu(B))=g_B(\alpha).$$ Therefore, $(g_X)_{X\in \mathcal{X}}$ is satisfies (a).  Let $X,Y\in \mathcal{X}$ be comonotonic and $\alpha\in \mathcal{D}(X)\cap \mathcal{D}(Y)$. By definition, there exists $x,y\in \mathbb{R}$ such that 
    $\alpha=\mathbb{Q}(X>x)=\mathbb{Q}(Y>y).$
    As $X$ and $Y$ are comonotonic, by \citet[Proposition 4.5]{D94A}, either $\{X>x\}\subseteq\{Y>y\}$ or $\{Y>y\}\subseteq \{X>x\}$. Since $\mathbb{Q}(X>x)=\mathbb{Q}(Y>y)$, it must hold that $\mathbb{Q}(\{Y>y\}\triangle \{X>x\})=0$. Therefore, by Proposition \ref{prop:nullity}, $$g_X(\alpha)=\tilde{g}_X(\alpha)=g(\nu(X>x))=g(\nu(Y>y))=\tilde{g}_Y(\alpha)=g_Y(\alpha),$$
    and $(g_X)_{X\in \mathcal{X}}$ satisfies (c). Thus $(g_X)_{X\in \mathcal{X}}$ is a distortion family. Furthermore, since $g_X(\alpha)=\tilde{g}_X(\alpha)$ holds for all $\alpha\in \mathcal{D}(X)$ and all $X\in \mathcal{X}$, it follows from \eqref{eq:gd} that for all $X\in \mathcal{X}(\mathcal{G})$, $g_X(\alpha)=g(\alpha)$ for all $\alpha\in \mathcal{D}(X)$. Therefore, $(g_X)_{X\in \mathcal{X}}$ is a $g$-distortion family.
    
    Denote by $u$ the $\succsim$-utility. Let $u':u(\mathbb{R})\to \mathbb{R}$ be the inverse of $u$. Given $X\in \mathcal{X}$ and $x\in u(\mathbb{R})$, by the definition of $\tilde{g}_X$, 
    $$g_X(\mathbb{Q}(u(X)>x))=g_X(\mathbb{Q}(X>u'(x)))=\tilde{g}_{X}(\mathbb{Q}(X>u'(x)))=g(\nu(X>u'(x)))=g(\nu(u(X)>x)).$$ 
    Therefore, $\int_{\Omega} u(X)\d(g\circ\nu)=\int_{\Omega} u(X)\d(g_X\circ\mathbb{Q})$ and $\succsim$ satisfies \eqref{eq:secondRep}.

    Let $(g_X)_{X\in \mathcal{X}}$ and $(\tilde{g}_X)_{X\in \mathcal{X}}$ be $g$-distortion families satisfying \eqref{eq:secondRep}. Fix $X\in \mathcal{X}$ and $\alpha\in \mathcal{D}(X)$. By definition, there exists $x\in \mathbb{R}$ such that $\alpha=\mathbb{Q}(X>x).$ Denote by $A=\{X>x\}$. Since $X$ and $\id_A$ are comonotonic and $\alpha\in \mathcal{D}(X)\cap \mathcal{D}(\id_A)$, $g_X(\alpha)=g_{A}(\alpha)$ and $\tilde{g}_X(\alpha)=\tilde{g}_{A}(\alpha).$ Let $c_A\in [0,1]$ satisfy $A\simeq c_A$,  by \eqref{eq:secondRep}, $$g_{A}(\alpha)=\int_\Omega u(\id_A) \d(g_{A}\circ \mathbb{Q})=u(c_A)=\int_\Omega u(\id_A) \d(\tilde{g}_{A}\circ \mathbb{Q})=\tilde{g}_{A}(\alpha).$$
Therefore, $g_{A}(\alpha)=\tilde{g}_{A}(\alpha)$ and $g_X(\alpha)=\tilde{g}_X(\alpha)$.

Assume that $\succsim$ satisfies Property (RAA). By Proposition \ref{prop:RAA}, $\mathbb{Q}\in \mathfrak{C}(\nu)$. Let $X\in \mathcal{X}$ and $\alpha\in \mathcal{D}(X)$. Find $x\in \mathbb{R}$ such that $\alpha=\mathbb{Q}(X>x)$. By \eqref{eq:def1}, we have
    $$g_X(\alpha)=\tilde{g}_X(\alpha)=g_X(\mathbb{Q}(X>x))=g(\nu(X>x))\leq g(\mathbb{Q}(X>x))=g(\alpha).$$
    Therefore, $(g_X)_{X\in \mathcal{X}}$ is ambiguity averse.

    For the converse, assume that $(g_X)_{X\in \mathcal{X}}$ is ambiguity averse. Let $X\in \mathcal{X}(\mathcal{G})$ and $Y\in \mathcal{X}$ satisfy $X=_{\mathbb{Q}}^{\mathrm{sd}}Y$, meaning, $\mathbb{Q}(X>x)=\mathbb{Q}(Y>x)$ for all $x\in \mathbb{R}$. Since $g_X(\alpha)=g(\alpha)$ for all $\alpha \in \mathcal{D}(X)$, and $g_Y(\alpha)\leq g(\alpha)$ for all $\alpha \in \mathcal{D}(Y)$, we have, for $x\in u(\mathbb{R})$,
    \begin{align*}
        g_X(\mathbb{Q}(u(X)>x))&=g_X(\mathbb{Q}(X>u'(x)))\\&=g(\mathbb{Q}(X>u'(x)))\\&=g(\mathbb{Q}(Y>u'(x)))\\&\geq g_Y(\mathbb{Q}(Y>u'(x)))=g_Y(\mathbb{Q}(u(Y)>x)),
    \end{align*}
    where $u':u(\mathbb{R})\to\mathbb{R}$ is the inverse of $u$. Thus, $\int_{\Omega}u(X)\d(g_X\circ \mathbb{Q})\geq \int_{\Omega}u(Y)\d(g_X\circ \mathbb{Q})$, and $X\succsim Y$. Therefore, $\succsim$ satisfies (RAA). 
\end{proof}

The following theorem confirms that \eqref{eq:secondRep} is also sufficient for $\succsim$ to be a CRDU preference relation.

\begin{theorem}
    \label{theo:3}
    Let $\succsim$ be a preference relation. Then $\succsim$ is a CRDU preference relation satisfying Property (NSC) if and only if there exist a continuous strictly increasing vNM utility function $u$, a continuous strictly increasing distortion function $g$, and a $g$-distortion family $(g_X)_{X\in \mathcal{X}}$ such that 
    \begin{equation}
    \label{eq:secondRep2}
        X\succsim Y \iff \int_{\Omega} u(X)\d(g_X\circ\mathbb{Q})\geq \int_{\Omega} u(Y)\d(g_Y\circ\mathbb{Q}),~~X,Y\in \mathcal{X}.
    \end{equation}
    Moreover, $u$ is unique up to positive affine transformations and can be taken as the $\succsim$-utility, $g$ is unique and given by the $\succsim$-risk distortion, and for each $X\in \mathcal{X}$, $g_X$ is uniquely defined on $\mathcal{D}(X)$. 
\end{theorem}

\begin{proof}
    The forward direction follows from Theorem \ref{theo:2}.
     Conversely, assume that there exist a continuous strictly increasing vNM utility $u$, a continuous strictly increasing distortion function $g$, and a $g$-distortion family $(g_X)_{X\in \mathcal{X}}$ such that 
    $$X\succsim Y \iff \int_{\Omega} u(X)\d(g_X\circ\mathbb{Q})\geq \int_{\Omega} u(Y)\d(g_Y\circ\mathbb{Q}),~~X,Y\in \mathcal{X}.$$
    
    Given $A\in \mathcal{F}$, define 
    $\tilde{\nu}(A)=g_A(\mathbb{Q}(A))$. Clearly, $\tilde{\nu}(\varnothing)=0$ and $\tilde{\nu}(\Omega)=1$. Let $A,B\in\mathcal{F}$ satisfy $A\subseteq B$, as $(g_X)_{X\in \mathcal{X}}$ satisfies (a), $$\tilde{\nu}(A)=g_A(\mathbb{Q}(A))\leq g_A(\mathbb{Q}(B))\leq g_B(\mathbb{Q}(B))=\tilde{\nu}(B).$$
    Therefore, $\tilde{\nu}(A)\leq \tilde{\nu}(B)$. We claim that $\tilde{\nu}$ is continuous. Let $(A_n)_{n\in\N}\subseteq\mathcal{F}$ be increasing and define $A=\bigcup_{n=1}^\infty A_n$ and $A_0=\varnothing$.
    Define $X=\sum_{n=1}^\infty 2^{-(n-1)}\id_{A_n\setminus A_{n-1}}$, it is straightforward to show  $$A_n=\left\{X>\frac{3}{4}\times 2^{-(n-1)} \right\},~~n\in \N,~~\text{and}~~A=\{X>0\}.$$
     Since $X$ and $\id_{A_n}$ are comonotonic for all $n\in \mathbb{N}$, $(g_X)_{X\in \mathcal{X}}$ satisfies (c), and $\mathbb{Q}(A_n)\in \mathcal{D}(X)\cap \mathcal{D}(\id_{A_n})$ for all $n\in \N$, it follows that $g_X(\mathbb{Q}(A_n))=g_{A_n}(\mathbb{Q}(A_n))=\tilde{\nu}(A_n)$ for all $n\in \N$. A similar argument will show that $g_X(\mathbb{Q}(A))=\tilde{\nu}(A)$. Therefore, as
     $(\mathbb{Q}(A_n))_{n\in \N}\subseteq \mathcal{D}(X)$ satisfies $\lim_{n\to\infty}\mathbb{Q}(A_n)=\mathbb{Q}(A)\in \mathcal{D}(X)$ and $(g_X)_{X\in \mathcal{X}}$ satisfies (b),    $$\lim_{n\to\infty}\tilde{\nu}(A_n)=\lim_{n\to\infty}g_X(\mathbb{Q}(A_n))=g_X(\mathbb{Q}(A))=\tilde{\nu}(A).$$
     Therefore, $\tilde{\nu}$ is upward continuous. A similar argument will show that $\tilde{\nu}$ is downward continuous. Therefore, $\tilde{\nu}$ is continuous.

    Since $g$ is continuous and strictly increasing, there exists a continuous strictly increasing distortion function $h$ such that $h(g(x))=g(h(x))=x$ for all $x\in [0,1]$. Define the capacity $\nu=h\circ \tilde{\nu}$. Clearly, $\nu$ is a continuous capacity. For all $A\in \mathcal{G}$, $g_A(\mathbb{Q}(A))=g(\mathbb{Q}(A))$ as $\mathbb{Q}(A)\in \mathcal{D}(\id_A)$. Therefore, as $\mathbb{Q}$ is risk conforming, we have 
    $$\nu(A)=h(\tilde{\nu}(A))=h(g_A(\mathbb{Q}(A)))=h(g(\mathbb{Q}(A)))=\mathbb{Q}(A)=\mathbb{P}(A),~~A\in \mathcal{G}.$$
    Thus, $\nu$ is risk conforming.
    
    Let $u':u(\mathbb{R})\to \mathbb{R}$ be the inverse of $u$. Given $X\in \mathcal{X}$ and $x\in u(\mathbb{R})$, it holds that $X$ and $\id_{\{X>u'(x)\}}$ are comonotonic and $\mathbb{Q}(X>u'(x))\in \mathcal{D}(X)\cap \mathcal{D}(\id_{\{X>u'(x)\}})$. Since $(g_X)_{X\in \mathcal{X}}$ satisfies (c), we have  
    \begin{align*}
        g_X(\mathbb{Q}(u(X)>x))=g_X(\mathbb{Q}(X>u'(x)))&=g_{\{X>u'(x)\}}(\mathbb{Q}(X>u'(x)))\\&=\tilde{\nu}(X>u'(x))=g(\nu(X>u'(x)))=g(\nu(u(X)>x)).
    \end{align*}
    Therefore,
    $$X\succsim Y\iff \int_{\Omega}u(X)\d(g\circ \nu)\geq \int_{\Omega}u(Y)\d(g\circ \nu),~~X,Y\in \mathcal{X}.$$
    Thus, $\succsim$ is a CRDU preference relation. 
    
    To show that $\succsim$ satisfies Property (NSC),  fix $A,B\in \mathcal{F}$ satisfying $\mathbb{Q}(A\triangle B)=0$. Therefore, $\mathbb{Q}(A\cap B)=\mathbb{Q}(A)$.
    Furthermore, we have that $\id_{A\cap B}$ and $\id_{A}$ are comonotonic and $\mathcal{D}(\id_{A\cap B})=\mathcal{D}(\id_A)=\{0,\alpha,1\}$, where $\alpha=\mathbb{Q}(A\cap B)=\mathbb{Q}(A)$. Therefore, since the distortion family $(g_X)_{X\in \mathcal{X}}$ satisfies (c),
    $$\nu(A\cap B)=g_{A\cap B}(\mathbb{Q}(A\cap B))=g_{A\cap B}(\alpha)=g_A(\alpha)=g_A(\mathbb{Q}(A))=\nu(A).$$
    Similary, one can show that $\nu(A\cap B)=\nu(B)$. Therefore, $\nu$ is $\mathbb{Q}$-consistent and, by Proposition \ref{prop:nullity}, $\succsim$ satisfies Property (NSC).

    The uniqueness claims are straightforward to show. 
\end{proof}

It is not immediately clear why \eqref{eq:secondRep2} implies that $\succsim$ satisfies Property (NSC). Therefore, we have the following proposition to assist in the sanity check of Theorem \ref{theo:3}.

\begin{proposition}
    \label{prop:sanity}
    Assume that $(g_X)_{X\in \mathcal{X}}$ is a distortion family. For all $X,Y\in \mathcal{X}$ satisfying $X=^{\mathrm{as}}_{\mathbb{Q}}Y$, it holds that $\mathcal{D}(X)=\mathcal{D}(Y)$ and $g_X(\alpha)=g_Y(\alpha)$ for all $\alpha\in \mathcal{D}(X)$.
\end{proposition}

\begin{proof}
    It is clear that $\mathcal{D}(X)=\mathcal{D}(Y)$. Let $\alpha\in \mathcal{D}(X)$, there exists $x\in \mathbb{R}$ such that $\alpha=\mathbb{Q}(X>x)$. Since $X=^{\mathrm{as}}_{\mathbb{Q}}Y$, it holds that $\alpha=\mathbb{Q}(Y>x).$ Let $A=\{X>x\}$ and $B=\{Y>x\}$. As $A\triangle B\subseteq \{X\neq Y\}$, $\mathbb{Q}(A\triangle B)=0$. Therefore, $\mathbb{Q}(A)=\mathbb{Q}(B)=\mathbb{Q}(A\cap B)=\alpha$. As $\id_{A\cap B}$ and $\id_{A}$ are comonotonic and $\mathcal{D}(\id_{A\cap B})=\mathcal{D}(\id_A)=\{0,\alpha,1\}$, by (c), we have
    $g_{A\cap B}(\alpha)=g_A(\alpha).$ Similarly, one can show that $g_{A\cap B}(\alpha)=g_B(\alpha)$. As $\id_A$ and $X$ are comonotonic and $\alpha\in \mathcal{D}(X)\cap \mathcal{D}(\id_A)$, by (c), we have $g_A(\alpha)=g_X(\alpha)$. Similarly, one can show that $g_B(\alpha)=g_Y(\alpha)$. Thus, by a chain of the above equalities, $g_X(\alpha)=g_Y(\alpha)$.
\end{proof}

\section{Proofs accompanying Section \ref{sec:dual}}
\label{app:dual}

\begin{proof}[Proof of Theorem \ref{theo:dual}]
    Let $\succsim$ be a preference relation satisfying Axioms (RC), (M), (SRM), (C), and (CI). By Proposition \ref{prop:cert}, for each $X\in \mathcal{X}$, there exists a unique $c_X\in \mathbb{R}$ satisfying $X\simeq c_X$. Define $V:\mathcal{X}\to \mathbb{R}$ by $V(X)=c_X$. Therefore,
    $$X\succsim Y\iff V(X)\geq V(Y), ~~X,Y\in \mathcal{X}.$$
    
    Let $X,Y\in \mathcal{X}$ be comonotonic. As $X,Y,V(X)$ are pairwise comonotonic, $Y,V(Y),V(X)$ are pairwise comonotonic, $X\simeq V(X)$, and $Y\simeq V(Y)$, by Axiom (CI), we have
    $$(1/2)X+(1/2)Y\simeq (1/2)V(X)+(1/2)Y\simeq (1/2)V(X)+(1/2)V(Y).$$
    Therefore, $V((1/2)X+(1/2)Y)=(1/2)V(X)+(1/2)V(Y)$. Taking $X=2Z$ (for $Z\in \mathcal{X}$) and $Y=0$, we see that $V(2Z)=2V(Z)$. Putting this all together, we get 
    $$V(X+Y)=V((1/2)2X+(1/2)2Y)=(1/2)V(2X)+(1/2)V(2Y)=V(X)+V(Y).$$
    Moreover, by Axiom (M), for all $X,Y\in \mathcal{X}$ satisfying $X\geq Y$, $V(X)\geq V(Y)$. As discussed in Section \ref{sec:def}, there exists a capacity $\nu$ such that $V(X)=\int_\Omega X\d\tilde{\nu}$ for all $X\in \mathcal{X}.$

    We claim that $\tilde{\nu}$ is continuous. Let $(A_n)_{n\in \N}\subseteq \mathcal{F}$ be an increasing sequence and define $A=\bigcup_{n\in \N}A_n$. Since $\id_{A_n}\to \id_A$ pointwise, by Proposition \ref{prop:cert},$$\lim_{n\to\infty}\nu(A_n)=\lim_{n\to\infty}\int_{\Omega}\id_{A_n}\d\tilde{\nu}=\lim_{n\to\infty}c_{\id_{A_n}}=c_{\id_A}=\nu(A).$$
    Therefore, $\tilde{\nu}$ is upward continuous. A similar argument will show that $\tilde{\nu}$ is downward continuous. Therefore $\tilde{\nu}$ is continuous. 

    Since $\succsim$ satisfies Axiom (RC), if $A,B\in \mathcal{G}$ satisfy $\mathbb{P}(A)=\mathbb{P}(B)$, then
    $$A\simeq B\implies V(\id_A)=V(\id_B)\implies \tilde{\nu}(A)=\tilde{\nu}(B).$$
    Thus, there exists a distortion function $g$ such that 
    $\tilde{\nu}(A)=g(\mathbb{P}(A))$ for all $A\in \mathcal{G}$. Since $\succsim$ satisfies Axiom (SRM), $g$ is strictly increasing. Since $\tilde{\nu}$ is continuous and $(\Omega,\mathcal{G},\mathbb{P})$ is atomless, $g$ is continuous. Let $h$ be the strictly increasing continuous distortion function that is the inverse of $g$. Define the capacity $\nu=h\circ \tilde{\nu}$. It is straightforward to show that $\nu$ is a continuous risk-conforming capacity satisfying $g\circ \nu=\tilde{\nu}$. 

    The converse is straightforward, and the uniqueness claims follow from Theorem \ref{theo:good}. 
\end{proof}

\section{Modularity results for lattices}
\label{ap:mod}

This appendix presents a technical result showing that on a general lattice, a distorted supermodular function is supermodular when the distortion function is increasing and convex. This result is used in the proof of Theorem \ref{prop:main}. 

A lattice is a partially ordered set $(\mathfrak{X},\leq)$ in which every pair $x,y\in\mathfrak{X}$ admits a least upper bound $x\vee y$ and a greatest lower bound $x\wedge y$. A function $\varphi:\mathfrak{X}\to\mathbb{R}$
is called \emph{supermodular} if for all $x,y\in\mathfrak{X}$, $\varphi(x\vee y)+\varphi(x\wedge y)\geq \varphi(x)+\varphi(y),$  \emph{submodular} if for all $x,y\in\mathfrak{X}$, $\varphi(x\vee y)+\varphi(x\wedge y)\leq \varphi(x)+\varphi(y),$ \emph{increasing} if for all $x,y\in\mathfrak{X}$ with $x\geq y$,
$\varphi(x)\geq \varphi(y),$ and \emph{decreasing} if for all $x,y\in\mathfrak{X}$ with $x\geq y$,
$\varphi(x)\leq \varphi(y).$

\begin{proposition}
    \label{prop:latt}
  Let $(\mathfrak{X},\leq)$ be a lattice, $a,b\in \mathbb{R}$ with $a<b$, $\varphi:\mathfrak{X}\to[a,b]$ be monotonic and supermodular (submodular), and $f:[a,b]\to\mathbb{R}$ be increasing and convex (concave). It holds that $g\circ\nu$ is supermodular (submodular).  
\end{proposition}
 
\begin{proof}
    Assume that $\varphi$ is supermodular and $f$ is convex. For $u,v\in \mathfrak{X}$, define
    $z=\varphi(u\wedge v)+\varphi(u\vee v)-\varphi(u)-\varphi(v)\geq 0.$
    Note that
    $$a\leq \varphi(v)+z=\varphi(u\wedge v)+\varphi(u\vee v)-\varphi(u)\leq \begin{cases}
        \varphi(u\vee v)&\text{if }\varphi\text{ is increasing}\\ \varphi(u\wedge v)&\text{if }\varphi\text{ is decreasing}
    \end{cases}\leq b.$$
    As $\varphi$ is monotonic, there exists $\lambda\in [0,1]$ such that
    $$\varphi(u)=\lambda \varphi(u\wedge v)+(1-\lambda)\varphi(u\vee v).$$
    Therefore,
    \begin{align*}
        f(\varphi(u\wedge v))+f(\varphi(u\vee v))&\geq f(\lambda\varphi(u\wedge v)+(1-\lambda)\varphi(u\vee v))+f((1-\lambda)\varphi(u\wedge v)+\lambda\varphi(u\vee v))\\&=f(\varphi(u))+f(\varphi(v)+z)\\&\geq f(\varphi(u))+f(\varphi(v)).
    \end{align*}
    A similar proof will show the claim when $\varphi$ is submodular and $f$ is concave.
\end{proof}

\end{document}